\documentclass[10pt,a4paper]{article}

\usepackage{amsmath,amssymb,amsthm,mathrsfs}
\usepackage{stmaryrd}  %\interleave
\usepackage{graphicx}
\usepackage{xcolor}
\usepackage{soul}
\usepackage{cancel}
\usepackage{changes}
\newcommand{\stkout}[1]{\ifmmode\text{\sout{\ensuremath{#1}}}\else\sout{#1}\fi}
\setdeletedmarkup{\stkout{#1}}
\usepackage{fullpage}

\newlength{\defbaselineskip}
\newcommand{\setlinespacing}[1]%
           {\setlength{\baselineskip}{#1 \defbaselineskip}}

\theoremstyle{plain}
\newtheorem{theorem}{Theorem}[section]
\newtheorem{lemma}[theorem]{Lemma}
\newtheorem{proposition}[theorem]{Proposition}

\theoremstyle{definition}

\newtheorem*{assumption}{Standing Assumption}

\theoremstyle{remark}
\newtheorem{remark}[theorem]{Remark}

\numberwithin{equation}{section}

 \allowdisplaybreaks
\begin{document}
\title{Long Time Behavior of Optimal Liquidation Problems\footnote{We thank Prof. Ulrich Horst and Prof. Xun Li for valuable comments and helpful discussion.}}

\author{  Xinman Cheng\footnote{The Hong Kong Polytechnic University, Department of Applied Mathematics, Kowloon, Hong Kong; email: xinman.cheng@connect.polyu.hk. X. Cheng would like to thank the Phd scholarship from The Hong Kong Polytechnic University.}\qquad Guanxing Fu\footnote{The Hong Kong Polytechnic University, Department of Applied Mathematics and Research Centre for Quantitative Finance, Kowloon, Hong Kong; email: guanxfu@polyu.edu.hk. G. Fu would like to thank financial supports through Hong Kong RGC (ECS) Grant No.25215122, NSFC Grant No. 12101523, and Research Centre for Quantitative Finance P0042708.} \qquad Xiaonyu Xia\footnote{Wenzhou University, College of Mathematics and Physics, Wenzhou 325035, PR China; email: xiaonyu.xia@wzu.edu.cn. X. Xia would like to thank financial support through NSFC Grant No. 12101465.}   }
%          ~ Paulwin Graewe\footnote{Department of Mathematics, Humboldt-Universit\"at zu Berlin,
%         Unter den Linden 6, 10099 Berlin, Germany; email: fuguanxi@math.hu-berlin.de} ~ Alexandre Popier\footnote{ Département de Mathématiques,
%     Laboratoire Manceau de Mathématiques, Université du Maine, Avenue Olivier Messiaen, 72085 Le Mans Cedex 9, France; Alexandre.Popier@univ-lemans.fr  }

\maketitle

\begin{abstract}
	In this paper, we study the long time behavior of an optimal liquidation problem with semimartingale strategies and external flows. To investigate the limit rigorously, we study the convergence of three BSDEs characterizing the value function and the optimal strategy, from finite horizon to infinite horizon.
We find that in the long time limit the player may not necessarily liquidate her assets at all due to the existence of external flows, even if in any given finite time horizon, the player is forced to liquidate all assets. Moreover, when the intensity of the external flow is damped, the player will liquidate her assets in the long run.
\end{abstract}



\section{Introduction and overview}\label{sec:intro}

	Let $W$ be a Brownian motion on a probability space $(\Omega,\mathbb{P})$ and $\mathbb{F}=
(\mathcal{F}_t)_{t\in [0,\infty)}$ be the augmented Brownian filtration. In this article, we study the long time behavior of the following stochastic control problem  
\begin{equation}\label{cost-T-continuous}
	J^{(T)}(X)=\mathbb{E}\left[\int_{0}^{T} e^{-\phi s}\left(-Y_{s-}\,dX_s + \frac{\gamma}{2}\,d[X]_s - \sigma_s \,d[X,W]_s+\lambda_s X_s^2\,ds \right) \right]\rightarrow \min \quad \textrm{over }X,
\end{equation}
subject to the state dynamics
\begin{equation}\label{state-T-continuous}
	\left\{ \begin{split}
		%	dX_s&=-dZ_s,\\
		dY_s&=-\rho_s Y_s\,ds-\gamma\,dX_s+\sigma_s\,dW_s,\\
		X_{0-}&=x_0,\quad X_T=0,\quad \,Y_{0-}=0.
	\end{split} \right.
\end{equation}		
	The stochastic control problem \eqref{cost-T-continuous}-\eqref{state-T-continuous} arises in optimal liquidation problems with semimartingale strategies. Here, $X_s$ is the position of the investor at time $s$ and $Y$ is the price deviation process that is widened by the investor's own trading $X$ (internal flow) and the external flow $\sigma\,dW$. The last term in \eqref{cost-T-continuous} is the risk aversion penalizing slow trading. The first three terms in \eqref{cost-T-continuous} are trading cost arising in the internal flow and the external flow. Here, $[\cdots]$ is the quadratic variation of $\cdots$. The model \eqref{cost-T-continuous}-\eqref{state-T-continuous} covers classic liquidation models in literature. When $\sigma=\phi=0$ and all coefficients are constant, \eqref{cost-T-continuous}-\eqref{state-T-continuous} reduces to \cite{OW-2013} if $X$ is a strategy of finite variation (singular control). When $\sigma=\phi=0$, \eqref{cost-T-continuous}-\eqref{state-T-continuous} is studied in \cite{GH-2017} with absolutely continuous strategies but random coefficients. In this paper, we assume $X$ to be a c\`adl\`ag semimartingale. 
	
	There are at least two motivating reasons for the appearance of semimartingale strategies. First, empirical studies illustrate the existence of quadratic variation of inventories. For example, Carmona and Webster \cite{carmona_webster} (see also Carmona and Leal \cite{Carmona-2023}) performed statistical test with the null hypothesis that the quadratic variation of inventories for high frequency traders is zero; the $p$-value is $10^{-5}$. 
	Thus, classical regular or singular controls are not appropriate in modeling. Second, optimal liquidation with semimartingale strategies appear naturally as the limit of the one with regular controls, when the temporary impact factor vanishes; see Horst and Kivman \cite{HKiv} for a rigorous proof of this convergence result.
	
	Portfolio liquidation problems with semimartingale strategies were first studied by Lorenz and Schied \cite{lorenz_schied}, where it is necessary to incorporate semimartingale strategies into the liquidation model when the benchmark price admits a drift. In \cite{GP}, G\^arleanu and Pedersen studied an infinite horizon stochastic control problem with semimartingale strategies arising in a portfolio choice problem; in \cite{GP}, semimartingale strategies appear naturally as the limit of discrete time model with only transient impact.
	For other portfolio liquidation problems using semimartingale strategies, refer to \cite{Ackermann2021,Ackermann2022,BBF-2018,FHX-2023,MuhleKarbe}. Among them, Becherer et al \cite{BBF-2018} studied an infinite horizon stochastic control problem with semimartingale strategies arising in a liquidation problem with multiplicative impact. In \cite{Ackermann2021}, Ackermann et al studied a finite horizon version of \eqref{cost-T-continuous} for a risk-neutral investor without external flow. In a following up result, the same authors (Ackermann, Kruse and Urusov) studied a generalized liquidation problem in \cite{Ackermann2021} with risk aversion and external flow; they extended the liquidation problem with semimartingale strategies to the one with progressively measurable strategies, which turned out to be equivalent to a linear quadratic control problem. Recently, Fu et al \cite{FHX-2023} introduce semimartingale strategies to the model in \cite{FHX1} with feedback order flow and solved a mean-field control problem. Mulhe-Karbe et al \cite{MuhleKarbe} investigated a nonlinear price impact model in discrete time and the scaling limit turned out to be a model with linear price impact, {and the optimal semimartingale strategy for the relevant optimal execution problem was obtained.}

Except for \cite{BBF-2018,GP}, all papers mentioned above studied finite horizon stochastic control problem. Moreover, \cite{BBF-2018,GP} studied models with constant coefficients.
	The goal of our paper is to study the long time behavior of the stochastic control problem \eqref{cost-T-continuous}-\eqref{state-T-continuous} with possibly random coefficients. Our contribution is two-fold. 
First, we prove rigorously the convergence from a finite horizon {\it constrained} stochastic control problem with random coefficients to an infinite horizon {\it unconstrained} stochastic control problem. The convergence of the control problems requires results of BSDE convergence as the horizon goes to infinity. {The convergence result of BSDEs is not new in literature, see e.g. \cite{Hu-Schweizer-2011,Hu-2022,peng2000infinite}. In constrast, our convergence result requires smaller solution space and finer estimate because of the structure of our BSDEs.} Moreover, our assumption on the risk aversion parameter $\lambda$ is mild enough such that our model allows both risk averse and risk neutral investors.
	As a byproduct of our convergence result, we solve the following infinite horizon stochastic control problem with c\`adl\`ag semimartingale strategies and possibly random coefficients
	\begin{equation}\label{cost-inf-continuous}
		J(X)=\mathbb{E}\left[\int_{0}^{\infty} e^{-\phi s}\left(-Y_{s-}\,dX_s + \frac{\gamma}{2}\,d[X]_s -\sigma_s \,d[X,W]_s+\lambda_s X_s^2\,ds \right) \right]\rightarrow \min \quad \textrm{over }X,
	\end{equation}
	subject to the state dynamics
	\begin{equation}\label{state-inf-continuous}
		\left\{ \begin{split}
			%	dX_s&=-dZ_s,\\
			dY_s&=-\rho_s Y_s\,dt-\gamma\,dX_s+\sigma_s\,dW_s,\\
			X_{0-}&=x_0,\quad Y_{0-}=0.
		\end{split} \right.
	\end{equation}	
Second, in the constant settings, we identify three properties of the optimal solution. The first property is that
 in the long run the investor does not necessarily liquidate her assets but instead her position exhibits fluctuations around zero\footnote{Note that in the random settings, the fluctuations can only be more significant.}. Such a phenomenon does not appear in literature on long time liquidation problems. For example, in \cite{BBF-2018} the investor would always liquidate all assets within a finite time period almost surely. The reason for the unusual behavior of the investor in our model is the existence of external flow $\sigma\,dW$; our investor needs to continuously trade to hedge the risk arising in the external flow who arrives in the investor's account continuously. An evidence for this interpretation is that the investor will liquidate all assets in the long run as long as the intensity of the external flow is weak enough. This is rigorously confirmed by the second property of the optimal solution, where we verify that once the position reaches zero, the trading sign of the investor is fully determined by that of the external flow; the investor will take advantage of the price movement by selling (buying) when others buy (sell). 
 %
% Note that the fluctuation is not because of market irregularity but a response of the investor to the continuous external flow.	
%
{Our study falls into the area of stochastic linear quadratic control problems in infinite horizon. In addition to \cite{HLY-2015} and \cite{SY-2018} for the study of optimal control with and without mean-field terms, we especially mention the result by Sun, Wang and Yong \cite{SWJ-2022}, where the turnpike property of the optimal state and control was established, namely, the expected state and strategy exhibit an exponetial decay to the respective stable points, as the horizon goes to infinity. In the third property, we verify that the optimal state and strategy in our problem also exhibit an exponential turnpike property. 
We distinguish ourselves from \cite{HLY-2015,SWJ-2022,SY-2018} by considering semimartingale strategies.}

	The rest of the paper is organized as follows. After introducing notation of spaces and standing assumption, we state the wellposedness result of the stochastic control problem with a finite horizon \eqref{cost-T-continuous}-\eqref{state-T-continuous} in Section \ref{sec:finite}. The convergence results are stated in Section \ref{sec:convergence}, including the convergence of BSDEs and the convergence of control problems. In Section \ref{sec:property} we identify three properties of the optimal solution, that is, the trading behavior in a long run, trading sign of the investor and the turnpike property.

	\textbf{Notation.}
	Let $\mathcal P(\mathcal S)$ be the space of all $\mathbb F\times[0,\infty)$-progressively measurable $\mathcal S$-valued stochastic processes.  
	
	For each $p>1$ and $K\in\mathbb R$, define
	\[
	\mathcal M^{p,K}(0,\infty;\mathcal S) = \left\{   v\in\mathcal P(\mathcal S):       \mathbb E\left[\left(\int_0^\infty e^{-2Kt} |v_t|^2\, dt\right)^{p/2}\right]<\infty 				\right\}. 
	\]
%	In particular, if $K=0$, we simply denote $\mathcal M^{p,0}(0,\infty;\mathcal S)$ by  $\mathcal M^{p}(0,\infty;\mathcal S)$.
	For each $p>1$ and $K\in\mathbb R$, define 
	\[
	S^{p,K}(0,\infty;\mathcal S)=\left\{    v\in\mathcal P(\mathcal S): 	\mathbb E\left[\sup_{t\geq 0}e^{-pKt}|v_t|^p\right]<\infty	   \right\}
	\]
	For each $p>1$ and $K\in\mathbb R$, define 
	\[
	L^{p,K}(0,\infty;\mathcal S)=\left\{  v\in\mathcal P(\mathcal S):      \mathbb E\left[\left(\int_0^\infty e^{-Kt} |v_t|\, dt\right)^p\right]<\infty     \right\}.
	\]
	In particular, if $\mathcal S=\mathbb R$, we will omit the dependence on $\mathcal S$ in above notation. 
	Note that the constant $K$ in the above definitions can be either positive or negative.
	
 In estimate, $c$ is a positive constant that may vary from line to line and that is independent of any horizon $T$.

	We assume throughout that the following {\bf Standing Assumption} holds.	
	\begin{assumption}\label{ass} 
		The price impact parameter $\gamma$ and the discount factor $\phi$ are positive constants\footnote{The constant assumption for $\phi$ is common in literature. The constant assumption for $\gamma$ is for simplicity. A general $\gamma$ would lead to quadratic BSDEs that complicate the analysis; see e.g. \cite{Ackermann2021}.}. 
		The volatility process $\sigma$ and resilience parameter $\rho$ are bounded progressively measurable processes. The risk aversion parameter $\lambda$ is a nonnegative progressively measurable stochastic process such that
		\begin{equation}\label{ass:lambda}
		\mathbb E\left[   \left( \int_0^\infty e^{ -\frac{1}{2}L s } \lambda_s \,ds\right)^2       \right]<\infty,
		\end{equation}
		for some $0<L<\phi.$ 
	\end{assumption}
For the $L$ in \eqref{ass:lambda}, the space of admissible strategies for both \eqref{cost-T-continuous}-\eqref{state-T-continuous} and \eqref{cost-inf-continuous}-\eqref{state-inf-continuous} is defined as 
	\begin{equation}\label{admissibility}
		\begin{split}
			\mathscr A:=&~\Bigg\{ X:  X \textrm{ is a c\`adl\`ag semimartingale such that }  \\
			&~\left.\qquad~  \mathbb E\left[ \sup_{0\leq s<\infty} e^{Ls}\widetilde X^4_s  \right] <\infty,\quad \mathbb E\left[  \left( \int_0^\infty e^{\frac{1}{2}Ls}\widetilde X^2_s\,ds    \right)^2    \right]<\infty   \right\},
		\end{split}
	\end{equation}
	%	\begin{equation}\label{admissibility}
		%	\begin{split}
			%		\mathscr A:=&~\left\{ Z:  \mathbb E\left[ \sup_{0\leq s<\infty} e^{Ks}\left(  \int_0^s e^{ -\frac{1}{2}\phi r }\,dZ_r		\right)^4\right] <\infty,\quad \mathbb E\left[  \left( \int_0^\infty e^{Ks}\widetilde Z^2_s\,ds    \right)^2    \right]<\infty,\right.\\
			%		&~\left. \qquad ~ \mathbb E\left[ \int_0^\infty d[\widetilde Z]_s   \right]<\infty,\quad \mathbb E\left[ \left(  \int_0^\infty e^{Ks-\frac{\phi}{2}s}\,d[\widetilde Z,W  ]_s  \right)^2\right]<\infty \right\},
			%	\end{split}
		%\end{equation}
		where $\widetilde X_s= e^{-\frac{\phi}{2}s}	X_s$.
	\begin{remark}
%\begin{itemize}
    To unify the notation, we consider \eqref{admissibility} as the space of admissible strategies for control problems with both finite and infinite horizons. Regarding the control problem on $[0,T]$, we consider its trivial extension by setting $\widetilde X_t=0$ for $t>T$.
%\item[(ii)] The assumption on $\lambda$ includes two special cases in literature. The first case is the one in Ankirchner and Kruse \cite{Ankirchner-2013} where $\lambda_t=\max\{0, c(a-S_t)\}$ with $a>0$ being a reference level, $c>0$ being the price sensitivity and $S$ being the benchmark price. The second case is $\lambda_t=\bar\lambda\bar\sigma^2_t$, where $\bar\sigma$ is the price volatility, including the Heston model $d\bar\sigma_t = \theta(w-\bar\sigma_t)\,dt+\xi\sqrt{ \bar\sigma_t }\,d\mathcal W_t$; here, $\mathcal W$ is a Brownian motion correlated with $W$.
%\end{itemize}
\end{remark}

%	\textcolor{red}{ 
%		\textrm{...A remark on the assumption on }$\gamma$\textrm{ is needed here or later... Is the relationship of }$Z^A$, $Z^B$ \textrm{ rely on that }$\gamma$ \textrm{ is a constant ?}
%	}

%%%%%%%%%%%%%%%%%%%%%%%%%%%%%%%%%%%
%%%%%%%%%%%%%%%%%%%%%%%%%%%%%%%%%%%
%%%%%%%%%%%%%%%%%%%%%%%%%%%%%%%%%%%
\section{Stochastic control problem with finite horizon}\label{sec:finite}
In order to study the long time limit of \eqref{cost-T-continuous}-\eqref{state-T-continuous}, in this section, we will solve the finite horizon control problem \eqref{cost-T-continuous}-\eqref{state-T-continuous} rigorously. In particular, we express the value function and the optimal state in terms of solutions to BSDEs on $[0,T]$. In the next section, we will study the long time limit of the value function and the optimal state by studying the convergence of the BSDEs as $T\rightarrow\infty$.

For each $t\in[0,T]$, we rewrite \eqref{cost-T-continuous}-\eqref{state-T-continuous} in the dynamic form 
\begin{equation}\label{cost-t-T}
	J^{(T)}(t, X)=\mathbb{E}\left[ \left.\int_{t}^{T} e^{-\phi s}\left(-Y_{s-}\,dX_s + \frac{\gamma}{2}\,d[X]_s -\sigma_s \,d[X,W]_s+\lambda_s X_s^2\,ds \right)\right|\mathcal F_t \right]\rightarrow \min\textrm{ over }\mathscr A,
\end{equation}
subject to the state dynamics
\begin{equation}\label{state-t-T}
	\left\{\begin{aligned}
		&~dY_s=-\rho Y_s\,ds-\gamma \,dX_s+\sigma_s\,dW_s,\quad s\in[t,T),\\
		&~X_{t-}=x,\quad X_T = 0, \quad Y_{t-}=y. 
	\end{aligned}\right.
\end{equation}
%where the set of admissible trading strategies is given by
%\begin{equation}
%\begin{split}
%	\mathscr A^{(T)}_t & := \left\{Z: Z \textrm{ is an } \mathbb{F} \textrm{ semimartingale with }\mathbb E\left[\int_t^T e^{-\phi s}d [Z]_s\Big|\mathcal F_t \right]<\infty.\right\}.
%\end{split}
%\end{equation}
For each $s\in[t,T]$, define 
\[ 
\widetilde{X}_s=e^{-\frac{\phi s}{2}}X_s,\qquad  \widetilde{Y}_s=e^{-\frac{\phi s}{2}}Y_s.
\] 
We can rewrite \eqref{cost-t-T} and \eqref{state-t-T} as follows:
\begin{equation}\label{cost-t-T-v2}
	\begin{split}
		J^{(T)}(t,\widetilde{X}) = \mathbb E\left[  \left. \int_t^T   \widetilde Y_{s-}\left(-\,d\widetilde X_s-\frac{\phi}{2}\widetilde{X}_s\,ds\right) + \frac{\gamma}{2}d[\widetilde X]_s - \sigma_s e^{-\frac{\phi}{2}s}\,d[\widetilde X,W]_s +\lambda_s \widetilde X^2_s    \,ds      \right|\mathcal F_t            \right],
	\end{split}
\end{equation}
respectively, 
\begin{equation}\label{state-t-T-v2}
	\left\{\begin{aligned}
		%		&~d\widetilde{X}_s=d\widetilde{X}_s,\\
		&~d \widetilde{Y}_s=\left(-\frac{\phi}{2}\widetilde{Y}_s-\rho \widetilde{Y}_s-\frac{\phi\gamma}{2}\widetilde{X}_s\right)\,ds-\gamma\,d \widetilde{X}_t+\sigma_s e^{-\frac{\phi s}{2}}\,dW_s,\quad s\in[t,T),\\
		&~\widetilde X_{t-}=e^{-\frac{\phi t}{2}}x,\quad \widetilde X_T = 0,\quad \widetilde Y_{t-}=e^{ - \frac{\phi t}{2}  }y. 
	\end{aligned}\right.
\end{equation}
For each $s\in[t,T]$, let $\mathcal X_s=(\widetilde{X}_s,\,\widetilde{Y}_s)^\top$. We can further rewrite \eqref{cost-t-T-v2}-\eqref{state-t-T-v2} in the following matrix form  
\begin{equation}\label{v2}
	J^{(T)}(t, \widetilde X)=\mathbb{E}\left[\left.\int_{t}^{T} \left(\mathcal X_{s}^\top\mathcal L d\widetilde{X}_s+ \mathcal R\,d[\widetilde{X}]_s +	\mathcal L^\top\mathcal{D}_{s} \,d[\widetilde{X},W]_s+\mathcal X_{s}^\top \mathcal Q_s\mathcal X_{s}\,ds \right)\right|\mathcal F_t \right],
\end{equation}
respectively,
\begin{equation}\label{v3}
	d \mathcal X_{s} =\mathcal{H}_s\mathcal X_s d s+\mathcal{D}_{s} d W_{s}+\mathcal{K} d \widetilde X_{s}, \quad s \in [t, T), 
\end{equation}
where
\begin{equation*}
	\begin{split}
		\mathcal L=&~\left(\begin{array}{lll}
			0 & -1 
		\end{array}\right)^{\top},\qquad  \mathcal R=\frac{\gamma}{2},\qquad 	\mathcal{D}_{s}=\left(\begin{array}{lll}
			0 & \sigma_{s}e^{-\frac{\phi s}{2}} 
		\end{array}\right)^{\top},\\ 
		\mathcal Q_s=&~\left(\begin{array}{ccc}
			\lambda_s & -\frac{\phi}{4}  \\
			-\frac{\phi}{4}  & 0
		\end{array}\right),\qquad  
		\mathcal{H}_s = \left(\begin{array}{ccc}
			0 & 0  \\
			-\frac{\phi}{2}\gamma & -\rho_s -\frac{\phi}{2}
		\end{array}\right), 
		\qquad 
		%	{ \bf I}_{2}=\left(\begin{array}{lll}
			%		1 & 0 \\
			%		0 & 1
			%	\end{array}\right) ,
		%	\quad 
		\mathcal{K}  =\left(\begin{array}{lll}
			1 & -\gamma 
		\end{array}\right)^{\top} .
	\end{split}
\end{equation*}
Given $\mathcal X:=\mathcal X_{t-}=(\widetilde{X}_{t-},\,\widetilde{Y}_{t-})$, the value function starting from $t$ is defined as 
\begin{equation*}\label{def:value-T}
	V^{(T)}(t, \mathcal X)=\inf_{\widetilde X \in  \mathscr A}J^{(T)}(t,\widetilde X).
\end{equation*}
Our goal is to represent the value function in terms of three stochastic processes $A^{(T)}, B^{(T)},C^{(T)}$ as
\begin{equation*}
	V^{(T)}(t,\mathcal X)=\mathcal X^\top A^{(T)}_t\mathcal X +\mathcal X^\top B^{(T)}_t+C^{(T)}_t.
\end{equation*}	
The dynamics of the processes $A^{(T)}, B^{(T)}, C^{(T)}$ is derived heuristically in Appendix \ref{sec:heuristics} by first analyzing a discrete time model and then taking the heuristic limit as the time difference between two consecutive trading periods tends to zero. It turns out that: 
\begin{itemize}
	\item The matrix-valued stochastic process $A^{(T)}=\begin{pmatrix}      A_{11}^{(T)}   &   A_{12}^{(T)} \\
		A_{21}^{(T)}   &   A_{22}^{(T)} 	      \end{pmatrix}$ is symmetric and satisfies $A^{(T)}_{11}=\gamma A^{(T)}_{21}$, $A^{(T)}_{12}=\gamma A^{(T)}_{22}+\frac{1}{2}$, and  $\overline A^{(T)} := A^{(T)}_{11}$ satisfies the following BSDE on $[0,T]$
	\begin{equation}\label{A-T}
		\left\{\begin{aligned}
			-d\overline{A}^{(T)}_{s}=&~-\left(\phi \overline A^{(T)}_{s}-\lambda_s+\frac{(\rho_s \overline A^{(T)}_{s}+\lambda_s)^2}{a_s}\right)\,ds-Z^{\overline A,(T)}_{s}\,dW_s,\\
			\overline A^{(T)}_{T}=&~ \frac{\gamma}{2}, 
		\end{aligned}\right. 
	\end{equation}
	where $a:=\gamma\rho+\lambda+\frac{\phi}{2}\gamma$, and the superscript $(T)$ emphasizes the dependence on $T$. 
	\item 	The vector-valued stochastic process $B^{(T)}=(B^{(T)}_1, B^{(T)}_2)^\top$ satisfies $B^{(T)}_1=\gamma B^{(T)}_2$, where for each $T$, $\overline B^{(T)}:=B^{(T)}_1$ satisfies the following BSDE on $[0,T]$
	\begin{equation}\label{B-T}
		\left\{\begin{aligned}
			-d\overline B^{(T)}_{s}=&~-\left\{\frac{\phi}{2}\overline B^{(T)}_{s}+\frac{(\lambda_s+\rho_s\overline A^{(T)}_{s})\rho_s \overline B^{(T)}_{s}}{a_s}-\frac{2}{\gamma}\sigma_se^{-\frac{\phi s}{2}}Z^{\overline A,(T)}_{s}\right\}\,dt-Z^{\overline B,(T)}_{s}\,dW_s,\\
			\overline B^{(T)}_{T}=&~0.
		\end{aligned}\right.
	\end{equation}
	\item The $\mathbb R$-valued stochastic process $C^{(T)}$ satisfies the following BSDE
	\begin{equation}\label{C-T}
		\left\{\begin{aligned}
			-d C^{(T)}_s=&~-\left\{-\sigma^2_se^{-\phi s}\frac{2\overline A^{(T)}_{s}-\gamma}{2\gamma^2_s}+\frac{\rho^2_s(\overline B^{(T)}_{s})^2}{4a_s}-\sigma_se^{-\frac{\phi s}{2}}\frac{Z^{\overline B,(T)}_{s}}{\gamma}\right\}\,ds-Z^{C,(T)}_s\,dW_s,\\
			C^{(T)}_T=&~0.
		\end{aligned}\right.
	\end{equation}
\end{itemize}
The wellposedness of $A^{(T)}$ can be found in \cite[Proposition 2.1]{Kohlmann2003}, the wellposedness of $B^{(T)}$ can be found in \cite[Theorem 4.1]{Briand03}, and the wellposedness of $C^{(T)}$ can be obtained by martingale representation theorem.

The proof of the following result is similar to \cite{FHX-2023}. We leave it to Appendix \ref{sec:verification-T} for completeness.
\begin{proposition}\label{thm:finite-control}
	Let the {\bf Standing Assumption} hold. 
	\begin{itemize}
		\item[i)] In terms of the processes $A^{(T)},B^{(T)},C^{(T)}$ introduced in \eqref{A-T}-\eqref{C-T}  the value function satisfies 
		\begin{equation} \label{value-T}
			V^{(T)}(t,\mathcal X)=\mathcal X^\top A^{(T)}_t\mathcal X +\mathcal X^\top B^{(T)}_t+C^{(T)}_t.
		\end{equation}
		\item[ii)] %The $e^{-\frac{1}{2}\phi\cdot}$-scaled optimal strategy $\widetilde X^{*,(T)}$  jumps only at the beginning and the end of the trading period where the initial and terminal jump is given by 
		Define $\widetilde X^{*,(T)}$ by 
		\begin{equation}\label{jump-T}
			\Delta {\widetilde X}^{*,(T)}_{t}=\frac{{I}_{t}^{A,(T)}}{a_t} {\mathcal X}_{t-}+\frac{{I}_{t}^{B,(T)}}{a_t}, \quad \mbox{and} \quad 
			\Delta \widetilde X^{*,(T)}_T =  -e^{-\frac{1}{2}\phi T} X^{*,(T)}_{T-},
		\end{equation}
	and 
		\begin{equation}\label{state-T-X}
		\widetilde X^{*,(T)}_s=\left(1-\frac{\lambda_s+\rho_s \overline A^{(T)}_{s}}{a_s}\right)  \widetilde P^{(T)}_s-\frac{\rho_s \overline B^{(T)}_{s}}{2a_s},  \qquad s\in(t,T),
	\end{equation}
		where %the processes $I^{A,(T)}$ and $I^{B,(T)}$ are given by 
		\begin{equation*}
			I^{A,(T)}=\left(\begin{matrix}
				-\rho \overline A^{(T)}-\lambda \\
				- \frac{\rho}{\gamma}\overline A^{(T)}+\rho +\frac{\phi}{2}
			\end{matrix}\right)^\top,
			\qquad 
			I^{B,(T)}=- \frac{\rho }{2}\overline B^{(T)},
		\end{equation*}
	and 
			\begin{equation*}
		\left\{\begin{aligned}
			d\widetilde P^{(T)}_s=&~-\left(\frac{\phi}{2}+\frac{\rho_s(\lambda_s+\rho_s \overline A^{(T)}_{s})}{a_s}\right)\widetilde P^{(T)}_s\,ds-\frac{ \rho^2_s \overline B^{(T)}_{s}}{2a_s}\,ds +\frac{\sigma_s}{\gamma} e^{-\frac{1}{2}\phi s} \,dW_s\\
			\widetilde P^{(T)}_0=&~x_0.
		\end{aligned}\right. 
	\end{equation*}
If $\widetilde X^{*,(T)}$ is a c\`adl\`ag semimartingale, then the minimal value \eqref{value-T} can be reached by using the strategy $\widetilde X^{*,(T)}$. In particular, the optimal strategy is unique in $\mathscr A$. In this situation, the optimal position is given by $X^{*,(T)}_\cdot:=e^{\frac{1}{2}\phi\cdot}\widetilde X^{*,(T)}_\cdot$, and the optimal deviation process is given by $Y^{*,(T)}_\cdot=e^{\frac{1}{2}\phi\cdot}\widetilde Y^{*,(T)}_\cdot$, where  
		\begin{equation}\label{state-T-Y}
		  \widetilde Y^{*,(T)}_s=\frac{\gamma(\lambda_s+\rho_s \overline A^{(T)}_{s})}{a_s} \widetilde P^{(T)}_s+\frac{\gamma\rho_s \overline B^{(T)}_{s}}{2a_s},\qquad s\in(t,T).
		\end{equation}

	\end{itemize}
\end{proposition}

%%%%%%%%%%%%%%%%%%
%%%%%%%%%%%%%%%%%%
%%%%%%%%%%%%%%%%%%

 \section{Convergence analysis}\label{sec:convergence}
In this section, we will prove that the value function in \eqref{value-T} converges to the one in the infinite time model \eqref{cost-inf-continuous}-\eqref{state-inf-continuous}. To do so, we will establish a convergence result of $(\overline A^{(T)},\overline B^{(T)},C^{(T)})$ as $T\rightarrow\infty$. To wit, we will prove $(\overline A^{(T)},\overline B^{(T)},C^{(T)})$ converges to the following BSDEs $(\overline A,\overline B,C)$ with infinite horizon:

\begin{itemize}
	\item The matrix-valued stochastic process $A$ is symmetric and satisfies $A_{11}=\gamma A_{21}$, $A_{12}=\gamma A_{22}+\frac{1}{2}$, where $\overline A:=A_{11}$ satisfies the following BSDE on $[0,\infty)$
	\begin{equation}\label{A-inf}
		-d\overline{A}_{t}=-\left(\phi \overline A_{t}-\lambda_t+\frac{(\rho_t \overline A_{t}+\lambda_t)^2}{a_t}\right)\,dt-Z^{\overline A}_{t}\,dW_t.
	\end{equation}
	\item 	The vector-valued stochastic process $B=(B_1,B_2)^\top$ satisfies $B_1=\gamma B_2$, where $\overline B:=B_1$ satisfies the following BSDE on $[0,\infty)$
	\begin{equation}\label{B-inf}
		-d\overline B_{t}=-\left\{\frac{\phi}{2}\overline B_{t}+\frac{(\lambda_t+\rho_t \overline A_{t})\rho_t \overline B_{t}}{a_t}-\frac{2}{\gamma}\sigma_te^{-\frac{\phi t}{2}}Z^{\overline A}_{t}\right\}\,dt-Z^{\overline B}_{t}\,dW_t.
	\end{equation}
	\item The $\mathbb R$-valued stochastic process $C$ satisfies the following BSDE
	\begin{equation}\label{C-inf}
		-d C_t=-\left\{-\sigma^2_te^{-\phi t}\frac{2\overline A_{t}-\gamma}{2\gamma^2}+\frac{\rho^2_t\overline B_{t}^2}{4a_t}  - \sigma_te^{-\frac{\phi t}{2}}\frac{Z^{\overline B}_{t}}{\gamma}\right\}\,dt-Z^C_t\,dW_t.
	\end{equation}
\end{itemize}
Although the wellposedness result for \eqref{A-inf} and the convergence result from $\overline A^{(T)}$ to $\overline A$ have been established in \cite{Hu-2022}, the standard $S^{2,0}\times\mathcal M^{2,0}$-estimate therein is not sufficient to establish the desired results including estimate and wellposedness for $\overline B$ and $C$, because of the dependence on $Z^{\overline A}$ and $(\overline B^2,Z^{\overline B})$ in the driver of $\overline B$ and $C$, respectively. 
In contrast, it requires us to establish finer $S^{p,K}$-estimate for $B^{(T)}$, which further requires a finer $\mathcal M^{p,K}$-estimate for $Z^{\overline A,(T)}$. For this reason, in Section \ref{sec:fine-estimate}, we will introduce an a priori $S^{p,K}\times\mathcal M^{p,K}$-estimate for a general BSDE with infinite horizon; this estimate will be frequently used in Section \ref{sec:convergence-ABC}, where we will establish the convergence result for \eqref{A-T}-\eqref{C-T} as well as the wellposedness result for \eqref{A-inf}-\eqref{C-inf}.

\subsection{Fine estimation for a general BSDE with infinite horizon}\label{sec:fine-estimate}
In this subsection, $|\cdot|$ stands for Euclidean norm of a matrix or vector, and $\langle \cdot, \cdot\rangle$ stands for the inner product between two vectors.
Consider the following infinite horizon BSDE
\begin{equation}\label{BSDE-app}
	-dY_t=(G(t,Y_t)+\varphi_t)dt-Z_t\,dW'_t,
\end{equation}
where %$$G(t,Y):[0,\infty)\times \mathbb R\rightarrow \mathbb R^d.$$	
 $W'$ is a $m$-dimensional Brownian motion. 
We assume that
\begin{itemize}
	\item (H1) for each $y\in\mathbb R^d,$ $G(\cdot,y)$ is a  progressively measurable process defined on $[0,\infty)$ with $G(t,0)\equiv0$, for each $t\in [0,\infty);$
	\item (H2) $G$ satisfies {\it weak monotonicity} condition, i.e., there exists some $\mu>0$ such that
	\[
	\big\langle G(t,y_1)-G(t,y_2),y_1-y_2 \big\rangle \leq -\mu|y_1-y_2|^2,\quad  \forall~ y_1,y_2\in\mathbb R^d,~t\in [0,\infty).
	\] 
%{\color{blue}	\item (H3) {$y\mapsto G(t,y)$ is Lipschitz.} delete!}
\end{itemize}	
The next proposition is adapted from \cite[Lemma 3.1 and Proposition 3.2]{Briand03}. Since the estimate therein does not depend on the horizon $T$, we can extend it to the case of infinite horizon, with an additional assumption on the terminal condition.
\begin{proposition}\label{prop:estimate1}
	Let $K\leq\mu$ and $p\geq2$ and Assumption (H1)-(H2) hold. Assume that $\varphi\in L^{p,K}_{}(0,\infty;\mathbb R^d)$. Let $(Y,Z)$ be a solution to BSDE \eqref{BSDE-app} such that $Y\in S^{p,K}_{ }(0,\infty;\mathbb R^d)$ and 
	\begin{equation}\label{boundary-requirement}
			\lim_{T\rightarrow\infty} \mathbb E[e^{-pKT}|Y_T|^p]=0. 
	\end{equation}
	Then, $Z\in \mathcal M^{p,K}(0,\infty;\mathbb R^{d\times m})$ and 
	\begin{equation}\label{estimate-main1}
		\begin{split}
			\mathbb E\left[\sup_{t\geq0}e^{-pKt}|Y_t|^p\right]+\mathbb E\left[\left(\int^\infty_0 e^{-2Kr}\mathrm{Tr}(Z^\top_rZ_r)\,dr\right)^{p/2}\right]&\leq c\mathbb E\left[\left(\int^\infty_0 e^{-Kr}|\varphi_r|\,dr\right)^p\right],
		\end{split}
	\end{equation}
where the positive constant $c$ only depends on $p$.
\end{proposition}

\begin{remark}
	%	1) From the proof we can see it is sufficient to have \eqref{boundary-requirement} up to a subsequence. 
	
	We only require $K\leq \mu$, and $K$ can be either positive or negative.
\end{remark}

For $i=1,2,$  consider the following two BSDEs:
\begin{equation}\label{BSDE-app1}
	-dY^i_t=(G^i(t,Y_t)+\varphi^i_t)dt-Z^i_t\,dW'_t.
\end{equation}
In a similar fashion as in Proposition \ref{prop:estimate1}, we have the following proposition, which will be used in Section \ref{sec:convergence-ABC} to establish Cauchy sequences in order to study convergence for \eqref{A-T}-\eqref{C-T}.
\begin{proposition}
\label{coro:Cauchy}
	  Let $K\leq \mu$ and $p\geq2$. For $i=1,2,$ assume that $G^i$ satisfies Assumption (H1)-(H2) and $\varphi^i\in L^{p,K}(0,\infty;\mathbb R^d).$ Let $(Y^i,Z^i)$ be the solution of BSDE \eqref{BSDE-app1} such that $Y^i\in S^{p,K}(0,\infty;\mathbb R^d)$ and $$\lim_{T\rightarrow\infty}\mathbb E\left[e^{-pKT}|	Y^i_T|^p\right]=0.$$  Then, $Z^i\in\mathcal M^{p,K}(0,\infty;\mathbb R^{d\times m})$, $i=1,2$, and there exists a constant $c>0$ such that 
	\begin{equation}\label{estimate-main2}
		\begin{split}
			&~\mathbb E\left[\sup_{t\geq0}e^{-pKt}|Y^1_t-Y^2_t|^p\right]+\mathbb E\left[\left(\int^\infty_0 e^{-2Kr}\mathrm{Tr}\left((Z^1_r-Z^2_r)^\top(Z^1_r-Z^2_r)\right)\,dr\right)^{p/2}\right]\\
			\leq&~ c \left(\mathbb E\left[\left(\int^\infty_0 e^{-Kr}|\varphi^1_r-\varphi^2_r|\,dr\right)^p\right]+\mathbb E\left[\left(\int^\infty_0 e^{-Kr}|G^1(r,Y^1_r)-G^2(r,Y^1_r)|\,dr\right)^p\right]\right).
		\end{split}
	\end{equation}
\end{proposition}

\begin{remark}

	If the driver $G$ is also Lipsichitz continuous in variable $y$, we can borrow the extension technique in \cite{peng2000infinite} together with Proposition \ref{prop:estimate1} and Proposition \ref{coro:Cauchy} to prove the wellposedness result for the BSDE \eqref{BSDE-app} in $S^{p,K}(0,\infty)\times \mathcal M^{p,K}(0,\infty)$; refer to Section \ref{sec:convergence-ABC}. The existence result for infinite horizon BSDEs has been considered in literature, but in different spaces, which are not sufficient to our purpose; see e.g. \cite{Briand03,peng2000infinite}. We need finer estimate to our purpose.
\end{remark}

%%%%%%%%%%%%%%%%%
%%%%%%%%%%%%%%%%%
\subsection{The convergence of $(\overline A^{(T)},\overline B^{(T)},C^{(T)})$}\label{sec:convergence-ABC}\label{sec:convergence-ABC}
In this subsection, we will establish a convergence result for the system \eqref{A-T}-\eqref{C-T} in appropriate spaces.
%\begin{equation}\label{eq:limit-ABC}
%	\left\{\begin{aligned}
%		-d\overline{A}_{s}=&~-\left(\phi \overline A_{s}-\lambda_s+\frac{(\rho_s \overline A_{s}+\lambda_s)^2}{a_s}\right)\,ds-Z^{\overline A}_{s}\,dW_s,\\
%		\textcolor{red}{\lim_{T\rightarrow\infty}\overline A_{T} } =&~ \frac{\gamma}{2}, \\
%		-d\overline B_{s}=&~-\left\{\frac{\phi}{2}\overline B_{s}+\frac{(\lambda_s+\rho_s\overline A_{s})\rho_s \overline B_{s}}{a_s}-\frac{2}{\gamma}\sigma_se^{-\frac{\phi s}{2}}Z^{\overline A}_{s}\right\}\,dt-Z^{\overline B}_{s}\,dW_s,\\
%		\textcolor{red}{\lim_{T\rightarrow\infty}\overline B_{T}}=&~0.\\
%		-d C_s=&~-\left\{-\sigma^2_se^{-\phi s}\frac{2\overline A_{s}-\gamma}{2\gamma^2_s}+\frac{\rho^2_s \overline B_{s}^2}{4a_s}-\sigma_se^{-\frac{\phi s}{2}}\frac{Z^{\overline B }_{s}}{\gamma}\right\}\,ds-Z^{C }_s\,dW_s,\\
%		\textcolor{red}{\lim_{T\rightarrow\infty}C_T}=&~0.
%	\end{aligned}\right.
%\end{equation}

\subsubsection{The convergence of $A^{(T)}$}
For each $t\geq 0$, denote 
$$
b_t =\frac{\rho_t^2}{\gamma\rho_t +\lambda_t+\frac{\phi}{2}\gamma}>0,\quad c_t=\frac{2\lambda_t \rho_t }{\gamma\rho_t +\lambda_t +\frac{\phi}{2}\gamma}+\phi>0,\quad d_t=\frac{\lambda^2_t}{\gamma\rho_t +\lambda_t +\frac{\phi}{2}\gamma}-\lambda_t<0.
$$
Note that $b$, $c$ and $d$ are bounded. The equation \eqref{A-T} can be rewritten as 
\begin{equation}\label{A-T-2}
	\left\{\begin{split}
		-d\overline A^{(T)}_{t}&= -\left(b_t(\overline A^{(T)}_{t})^2+c_t\overline A^{(T)}_{t}+d_t\right)\,dt-Z^{\overline A,(T)}_t\,dW_t,\\
		\overline A^{(T)}_{T}&=\frac{\gamma}{2}.
	\end{split}\right. 
\end{equation}
Set $\widehat{A}^{(T)}:=\overline A^{(T)} -\frac{\gamma}{2}$. It follows that
\begin{equation}\label{eq:A-hat}
	\left\{\begin{split}
		-d{\widehat{A}}^{(T)}_{t}=&~\left(  G(t, \widehat{A}^{(T)}_{t})+\varphi_t    \right)\,dt - Z^{\widehat A,(T)}_t\,dW_t, \\
		\widehat A^{(T)}_{T}=&~ 0,
	\end{split}\right. 
\end{equation}
where $G(t, A)=-b_t A^2-(b_t\gamma + c_t)A $ and $\varphi_t=-\frac{\gamma^2}{4}b_t-\frac{\gamma}{2}c_t-d_t$.

Let $\{T_k\}_{k}$ be any sequence that satisfies $T_k\nearrow \infty$ as $k\rightarrow \infty$. We set 
$$\varphi^k_t=1_{[0,T_k]}(t)\varphi_t, \quad t\in [0,\infty).$$
Consider for each $k$ the following infinite horizon BSDEs
\begin{equation}\label{eq:truncated-BSDE-infinite}
\begin{aligned}-d\widehat A^{k}_{t}=&\left(G(t,\widehat A^k_{t})+\varphi^k_t\right)\,dt-Z^k_t \,dW_t,\quad& t\geq 0.
\end{aligned}
\end{equation}
Next we solve the BSDE \eqref{eq:truncated-BSDE-infinite} and establish an estimate for solutions in $S^{p,K}(0,\infty)\times \mathcal M^{p,K}(0,\infty)$.
\begin{lemma}\label{prop:Ak-infinite}
	For every $0<K\leq \phi$ and $p\geq2$, the infinite horizon BSDE \eqref{eq:truncated-BSDE-infinite} admits a solution $(\widehat A^k,Z^{\widehat A,k})\in  S^{p,K}(0,\infty)\times \mathcal M^{p,K}(0,\infty)$ such that $\widehat A^k$ is $[-\frac{\gamma}{2},0]$-valued, and the solution is unique in the space $S^{p,K}(0,\infty;[-\gamma/2,0])\times \mathcal M^{p,K}(0,\infty)$.
 
 Moreover, we have the following estimate	
	\begin{equation}\label{estimate-ZAk}
	\mathbb E\left[\sup_{t\geq0}e^{-pKt}|\widehat A^k_t|^p\right]+\mathbb E\left[\left(\int^\infty_0 e^{-2Kr}|Z^{\widehat A,k}_r|^2\,dr\right)^{p/2}\right]\leq c\mathbb E\left[ \left(\int^\infty_0 e^{-Kr}|\varphi_r|\,dr\right)^p\right]<\infty,
	\end{equation}
where $c$ is independent of $k$. 
\end{lemma}
\begin{proof}
	{\bf Step 1.} 
In this step, we verify the uniform boundedness of $\widehat A^{k}$. To do so, we borrow the truncation argument from \cite{Ackermann2021}. 
	
	We first consider the BSDE \eqref{eq:A-hat}  with $T=T_k$ and a truncated driver $f(t,A)= G\left(t, L(A)\right)+\varphi_t, $
	where $L(A)=(A \lor (-\frac{\gamma}{2}))\land 0$.
	We can find a constant $c>0$ such that 
	$$|  f(s,A)-  f(s,A')|\leq c|A-A'|,\quad  \forall s\in[0,T],$$
	which implies the existence of a unique solution $(\widehat A^k,Z^{\widehat A,k})\in  S^{2,0}(0,T_k)\times \mathcal M^{2,0}(0,T_k)$ to the BSDE with the driver $f$.  
	
	Next we are going to show that $\widehat{A}^{k}$ is $[-\frac{\gamma}{2}, 0]$-valued. Note that $(-\frac{\gamma}{2},0)$ (respectively, $(0,0)$) solves the BSDE
	$d\widehat{A}^{k}_{t}=Z^{\widehat A,k}_t\,dW_t$
	with terminal condition $-\frac{\gamma}{2}$ (respectively, $0$).
	Moreover, we can check that
	$$ {f}\left(t,-\frac{\gamma}{2}\right)= -d_t>0,$$
	and that
	$$ {f}(t,0)=-b_t\left(\frac{\gamma}{2}-\frac{c_t}{2b_t}\right)^2+\frac{c_t^2}{4a_t}-d_t=-\frac{\gamma^2\rho_t^2+2\phi\rho_t\gamma^2+\phi^2\gamma^2}{4(\gamma\rho_t+\lambda_t+\frac{\phi}{2}\gamma)}<0.
	$$
	By the comparison principle \cite[Proposition C.1.]{Ackermann2021}, we have that  $\widehat{A}^{k}$ is $[-\frac{\gamma}{2}, 0]$-valued. Now we trivially extend this unique solution to $[0,\infty)$ by letting $(\widehat A^k,Z^{\widehat A,k})=(0,0)$ for all $t>T_k$. Thus, we get a bounded solution $(\widehat A^k,Z^{\widehat A,k})\in  S^{2,0}(0,\infty)\times \mathcal M^{2,0}(0,\infty)$ to the BSDE \eqref{eq:truncated-BSDE-infinite}. 
 
	{\bf Step 2.} Since $\widehat{A}^{k}$ and $\varphi^k$ are uniformly bounded, we have that $\widehat{A}^{k}\in  S^{p,K}(0,\infty)$ and $\varphi^k\in L^{p,K}(0,\infty)$ for every $p\geq2$ and $K>0$. Moreover, it can be verified that $G$ satisfies weak monotonicity with $y_1,y_2\in[-\frac{\gamma}{2},0]$:
	\begin{equation}\label{eq:weak-monotonicity-A}
	\begin{split}
	&~\Big( G\left(t,y_1\right)-G\left(t,y_2\right)\Big)\left(y_1-y_2\right) \\
	=&~\Big(-b_t(y_1+y_2)(y_1-y_2)-(b_t\gamma+c_t)(y_1-y_2) \Big) ( y_1-y_2 )  \\
	\leq&~(b_t\gamma-b_t\gamma-c_t)|y_1-y_2|^2\\
	\leq&~-\phi|y_1-y_2|^2.
	\end{split}   
	\end{equation}
	Besides, we have that $\lim\limits_{T\rightarrow \infty}\mathbb E\left[ e^{-pKT} |\widehat{A}^{k}_T|^p \right]=0$. Thus, the requirements in Proposition \ref{prop:estimate1} are satisfied. Then we have   $Z^{\widehat A,k}\in \mathcal M^{p,K}(0,\infty)$ and   the estimate \eqref{estimate-ZAk} for every $p\geq2$ and $0<K\leq \phi$ .

 For any two solutions to \eqref{eq:truncated-BSDE-infinite} in the space $S^{p,K}(0,\infty;[-\gamma/2,0])\times \mathcal M^{p,K}(0,\infty)$, the weak monotonicity condition \eqref{eq:weak-monotonicity-A} holds. Thus, the uniqueness result follows from Proposition \ref{coro:Cauchy}.
\end{proof}
The next theorem shows the convergence result of \eqref{eq:truncated-BSDE-infinite} as well as a fine estimate for the limit, which will be used in the convergence of $\overline B^{(T)}$ in Section \ref{sec:convergence-B}. 
\begin{theorem}\label{thm:A}
	Let $p=8$ and $K=\frac{\phi}{8}$ in Lemma \ref{prop:Ak-infinite}. The sequence of solutions $\{(\widehat{A}^{k},Z^{\widehat{A},k})\}_k$ of the BSDE \eqref{eq:truncated-BSDE-infinite} admits a limit in $S^{8,\frac{\phi}{8}}(0,\infty)\times \mathcal M^{8,\frac{\phi}{8}}(0,\infty)$. Denote the limit by $(\widehat A,Z^{\widehat A})$, which is the solution to the infinite horizon BSDE
\begin{equation}\label{eq:A-hat-infty}
-d{\widehat{A}} _{t}=\left(G(t,\widehat A _{t})+\varphi_t\right)\,dt-Z^{\widehat A}_t\, dW_t,\qquad  t\in[0,\infty).
\end{equation}
		Moreover, we have the following estimate for $Z^{\widehat A}$
	\begin{equation}\label{estimate-ZA}
		\mathbb E\left[\left(\int^\infty_0 e^{-\frac{\phi r}{4}}|Z^{\widehat A}_r|^2\,dr\right)^4\right]\leq c\mathbb E\left[ \left(\int^\infty_0 e^{-\frac{\phi r}{8}}|\varphi_r|\,dr\right)^8\right]<\infty.
	\end{equation}	
The solution $(\widehat A,Z^{\widehat A})$ is unique in the space $S^{8,\frac{\phi}{8}}(0,\infty;[-\gamma/2,0])\times \mathcal M^{8,\frac{\phi}{8}}(0,\infty)$.
\end{theorem}
\begin{proof}
		For any $m>n$, we denote by $(\widehat{A}^{m},Z^{\widehat{A},m})$ and $ (\widehat{A}^{n},Z^{\widehat{A},n})$ the solutions to the BSDEs \eqref{eq:truncated-BSDE-infinite} with $k=m$ and $n$, respectively. It can be easily checked that $\{\varphi^k\}$ converges to $\varphi_t$ in $L^{8,\frac{\phi}{8}}_{}(0,\infty)$. From Proposition \ref{coro:Cauchy}, we have that
			\begin{equation}\label{estimate-cauchyA}
		\begin{split}
		&~\mathbb E\left[\sup_{t\geq0}e^{-\phi t}|\widehat{A}^{m}_t-\widehat{A}^{n}_t|^8\right]+\mathbb E\left[\left(\int^\infty_0 e^{-\frac{\phi r}{4}}|Z^{\widehat{A},m}_r-Z^{\widehat{A},n}_r|^2\,dr\right)^{4}\right]\\
		\leq&~ c\mathbb E\left[\left(\int^\infty_0 e^{-\frac{\phi r}{8}}|\varphi^m_r-\varphi^n_r|\,dr\right)^8\right].
		\end{split}
		\end{equation}
		Hence, we know that $\{(\widehat{A}^{k},Z^{\widehat{A},k})\}_k$ is a Cauchy sequence in $S^{8,\frac{\phi}{8}}(0,\infty)\times \mathcal M^{8,\frac{\phi}{8}}(0,\infty)$. It is straightforward to check that the limit $(\widehat A,Z^{\widehat A})\in    S^{8,\frac{\phi}{8}}(0,\infty;\mathbb R)\times \mathcal M^{8,\frac{\phi}{8}}(0,\infty;\mathbb R)$ solves the infinite horizon BSDE \eqref{eq:A-hat-infty} and that $\widehat A$ is valued in $[-\frac{\gamma}{2},0]$. Finally, by $\lim_{T\rightarrow\infty}	\mathbb E\left[   e^{-\phi T}| \widehat A_T |^8    \right]=0$ and Proposition \ref{prop:estimate1}, we can obtain the estimate \eqref{estimate-ZA}. 

 Uniqueness follows from weak monotonicity of $G$ for bounded arguments and Proposition \ref{coro:Cauchy}.
\end{proof}
Denote $\overline A^{k}=\widehat A^{k}+\frac{\gamma}{2}$, $Z^{\overline A,k}=Z^{\widehat A,k}$,  $\overline A=\widehat A+\frac{\gamma}{2}$ and $Z^{\overline A}=Z^{\widehat A}$. We have the same convergence result from $(\overline A^{k},Z^{\overline A,k})$ to $(\overline A,Z^{\overline A})$ in $S^{8,\frac{\phi}{8}}(0,\infty)\times \mathcal M^{8,\frac{\phi}{8}}(0,\infty)$ . By the arbitrary choice of the sequence $\{T_k\}$ and the uniqueness of solution of the BSDE \eqref{A-inf} in the space $S^{8,\frac{\phi}{8}}(0,\infty;[-\gamma/2,0])\times \mathcal M^{8,\frac{\phi}{8}}(0,\infty)$, we deduce that the solution of the BSDE \eqref{A-T} converges to the solution of the BSDE in \eqref{A-inf} as $T\rightarrow \infty.$ 
%	\textcolor{red}{ 
	%	\begin{remark}
		%		Theorem \ref{thm:A} holds for $K=\phi$ but in the convergence result of $C$ in Section \ref{sec:convergence-C} we need to require $K<\phi$. 
		%	\end{remark}
	%}
%%%%%%%%%%%%%%%%%%%%%
%%%%%%%%%%%%%%%%%%%%%
\subsubsection{The convergence of $B^{(T)}$}\label{sec:convergence-B}

 Consider the infinite horizon BSDE 
\begin{equation}\label{B-trivial-infinite}
-d\overline B^{k}_{t}=~\left\{\widetilde{G}^k(t, \overline B^{k}_t)+\widetilde{\varphi}^{k}_t\right\}\,dt - Z^{\overline B,k}_{t}\,dW_t,\quad t\geq 0
\end{equation}
where $\widetilde G^k(t,B)=-\left(\frac{\phi}{2}+\frac{(\lambda_t+\rho_t \overline A^{k}_{t})\rho_t}{\gamma\rho_t+\frac{\phi}{2}\gamma+\lambda_t}\right)B$ and  $\widetilde{\varphi}^{k}_t=\frac{2}{\gamma}\sigma_te^{-\frac{\phi t}{2}}Z^{\overline A,k}_t.$
Note that for $t>T_k$, $\overline A^{k}_t =\frac{\gamma}{2}$ and $Z^{\overline A,k}_t=Z^{\widehat A,k}_t=0$. In order to obtain a solution to \eqref{B-trivial-infinite}, we only need to solve the following BSDE on $[0,T_k]$  
\begin{equation}\label{B-trivial-finite}
\left\{\begin{split}
-d\overline B^{k}_{t}=&~\left\{\widetilde{G}^k(t, \overline B^{k}_t)+\widetilde{\varphi}^{k}_t\right\}\,dt - Z^{\overline B,k}_{t}\,dW_t,\quad &t\in[0,T_k), \\
\overline B_{T_k}^{k} =&~ 0,
\end{split}\right.
\end{equation}
and then let $\overline B^{k}=Z^{\overline B,k}=0$ on $(T_k,\infty).$ The next lemma provides an a priori estimate for $\overline B^k$ and $Z^{\overline B,k}$, which will be used in the convergence result in Theorem \ref{thm:B}. 
\begin{lemma}\label{prop:Bk}
	The infinite horizon BSDE \eqref{B-trivial-infinite} admits a unique solution $(\overline B^k,Z^{\overline B,k})\in  S^{p,K}(0,\infty)\times\mathcal M^{p,K}(0,\infty)$.	Moreover, for each $-\frac{\phi}{2}< K \leq\frac{\phi}{2}$ and $p\geq2,$ we have the following estimate	
	\begin{equation}\label{estimate-ZBk}
	\mathbb E\left[\sup_{t\geq0}e^{-pKt}|\overline B^k_t|^p\right]+\mathbb E\left[\left(\int^\infty_0 e^{-2Kr}|Z^{\overline B,k}_r|^2\,dr\right)^{p/2}\right]\leq c \mathbb E\left[ \left(\int^\infty_0 e^{-(K+\frac{\phi}{2})r} | Z^{\overline A,k}_r  |^2\,dr     \right)^{p/2}\right]<\infty.
	\end{equation}
\end{lemma}
\begin{proof}
		Since $\overline A^{k}$ is $[0,\frac{\gamma}{2}]$-valued and $\rho$ is bounded, we obtain that 
	\[
	\Big(  \widetilde G(\cdot,y_1)-\widetilde G(\cdot,y_2)\Big)  \Big(   y_1-y_2\Big) \leq -\frac{\phi}{2}|y_1-y_2|^2
	\]
	and that 
	\[ 
	|\widetilde G(\cdot,y_1)-\widetilde G(\cdot,y_2)|\leq c|y_1-y_2|.
	\]
	Moreover, for any $p\geq2, K\in(-\frac{\phi}{2},\frac{\phi}{2}]$ , we have that
		\begin{equation*}
	\begin{split}
	~\mathbb E\left[  \left( \int_0^\infty e^{ -Kr} | \widetilde\varphi^k_r|\,dr   \right)^p   \right]\leq &~ c\mathbb E\left[  \left(   \int_0^\infty  e^{-(K+\frac{\phi}{2})r} |Z^{\overline A,k}_r| \,dr     \right)^p      \right] \\
	\leq&~c\mathbb E\left[ \left(  \int_0^\infty e^{-(K+\frac{\phi}{2})r} | Z^{\overline A,k}_r  |^2\,dr     \right)^{p/2}  \right]   \qquad (\textrm{by H\"older's inequality})\\
	<&~\infty   \qquad\qquad  (\textrm{by }\eqref{estimate-ZAk}),
	\end{split}
	\end{equation*}
	which implies that $\widetilde\varphi^k\in  L^{p,K}_{}(0,\infty)$. From \cite[Theorem 4.2]{Briand03}, we know that there exists a unique solution $(\overline B^k,Z^{\overline B,k})$ to the linear BSDE \eqref{B-trivial-finite} in $S^{p,K}_{ }(0,T_k)\times\mathcal M^{p,K}_{ }(0,T_k)$. Letting $\overline B^k=Z^{\overline B,k}=0$ on $(T_k,\infty)$, we have that the processes $(\overline B^k,Z^{\overline B,k})\in S^{p,K}(0,\infty)\times\mathcal M^{p,K}(0,\infty)$ solve the infinite horizon BSDE \eqref{B-trivial-infinite}. Together with the fact that $\lim\limits_{T\rightarrow \infty}\mathbb E\left[ e^{-pKT} |\overline B^{k}_T|^p \right]=0$, we obtain the required estimate \eqref{estimate-ZBk} by Proposition \ref{prop:estimate1}. 
\end{proof}
The next theorem shows the convergence result from \eqref{B-trivial-infinite} to \eqref{B-inf}.
\begin{theorem}\label{thm:B}
		Let $p=4$ and $K=-\frac{\phi}{4}$ in Lemma \ref{prop:Bk}. The sequence of solutions $\{(\overline B^{k},Z^{\overline B,k})\}_k$ admits a limit in $S^{4,-\frac{\phi}{4}}(0,\infty)\times \mathcal M^{4,-\frac{\phi}{4}}(0,\infty)$, denoted by $(\overline B,Z^{\overline B})$, which is the unique solution to the infinite horizon BSDE \eqref{B-inf}. In addition, we have the following estimate
\begin{equation}\label{estimate-ZB}
\mathbb E\left[\sup_{t\geq0}e^{\phi t}|\overline B_t|^4\right]+\mathbb E\left[\left(\int^\infty_0 e^{\frac{\phi r}{2}}|Z^{\overline B}_r|^2\,dr\right)^{2}\right]\leq c\mathbb E\left[ \left(  \int_0^\infty e^{-\frac{\phi}{4}r} | Z^{\overline A}_r  |^2\,dr     \right)^2   \right]<\infty.
\end{equation}	
\end{theorem}
\begin{proof}
For any $m>n$, we denote by $(\overline B^{m},Z^{\overline B,m})$ and $(\overline B^{n},Z^{\overline B,n})$ the solutions to the BSDEs \eqref{B-trivial-finite} with $k=m$ and $n$, respectively. By Proposition \ref{coro:Cauchy}, we get that
\begin{equation*}
	\begin{split}
		&~\mathbb E\left[\sup_{t\geq0}e^{\phi t}|\overline B^{m}_t-\overline B^{n}_t|^4\right]+\mathbb E\left[\left(\int^\infty_0 e^{\frac{\phi r}{2}}|Z^{\overline B,m}_r-Z^{\overline B,n}_r|^2\,dr\right)^{2}\right]\\
		\leq&~ c\left(\mathbb E\left[\left(\int^\infty_0 e^{\frac{\phi r}{4}}|\widetilde\varphi^{m}_r-\widetilde\varphi^{n}_r|\,dr\right)^4\right]+\mathbb E\left[\left(\int^\infty_0 e^{\frac{\phi r}{4}}|(\overline A^{m}_{r}-\overline A^{n}_{r})\overline B^{m}_r|\,dr\right)^4\right]\right)\\
		\leq&~ c\left(\mathbb E\left[ \left( \int_0^\infty e^{-\frac{\phi r}{4}}  \left| Z^{\overline A,m}_r- Z^{\overline A,n}_r        \right|^2 \,dr \right)^4    \right]^{1/2}+\mathbb E\left[\left(\sup_{t\geq0}e^{-\frac{\phi t}{8}}|\overline A^{m}_{t}-\overline A^{n}_{t}|\sup_{t\geq0}e^{\frac{7\phi t}{16}}|\overline B^{m}_t|\int^\infty_0 e^{-\frac{\phi r}{16}}\,dr\right)^4\right]\right)\\
			\leq&~ c\left(\mathbb E\left[ \left( \int_0^\infty e^{-\frac{\phi r}{4}}  \left| Z^{\overline A,m}_r- Z^{\overline A,n}_r        \right|^2 \,dr \right)^4    \right]^{1/2}+\mathbb E\left[\sup_{t\geq0}e^{-\phi t}|\overline A^{m}_{t}-\overline A^{n}_{t}|^8\right]^{1/2}\cdot\mathbb E\left[\sup_{t\geq0}e^{\frac{7\phi t}{2}}|\overline B^m_t|^8\right]^{1/2}\right).
	\end{split}
\end{equation*}
Recalling the estimate \eqref{estimate-ZBk} with $K=-\frac{7\phi}{16}$ and $p=8$, we have
\begin{equation*}
	\begin{split}
\sup_{m}\mathbb E\left[\sup_{t\geq0}e^{\frac{7\phi t}{2}}|\overline B^m_t|^8\right]\leq c \sup_m\mathbb E\left[ \left(\int^\infty_0 e^{-\frac{\phi r}{16}} | Z^{\overline A,m}_r  |^2\,dr\right)^{4}\right]<\infty,
	\end{split}
\end{equation*}
where the last inequality is obtained by the estimate \eqref{estimate-ZAk} by noting that the constant $c$ in \eqref{estimate-ZAk} is independent of $k$.

By applying Proposition \ref{coro:Cauchy} to $(\overline A^k,Z^{\overline A,k} )$ (refer to the estimate \eqref{estimate-cauchyA}), $\{\overline B^{k}, Z^{\overline B,k} \}_k$ is a Cauchy sequence in $S^{4,-\frac{\phi}{4}}(0,\infty)\times \mathcal M^{4,-\frac{\phi}{4}}(0,\infty)$. It is easy to show that the limit $(\overline B,Z^{\overline B})$ satisfies the BSDE \eqref{B-inf}. Finally, note that
\[
\mathbb E\left[ e^{  \phi T  }|\overline B_T|^4  \right]\leq 	c\left(\mathbb E\left[ \sup_{t\geq 0}e^{  \phi t  }|\overline B^{k}_t-\overline B_t|^4  \right] + 	\mathbb E\left[ e^{  \phi T  }|\overline B^{k}_T|^4  \right]\right).
\]
Letting $T\rightarrow\infty$ first and then $k\rightarrow\infty$ we have $\lim\limits_{T\rightarrow\infty}	\mathbb E\left[ e^{  \phi T  }|\overline B_T|^4  \right]=0$. Thus, 
Proposition  \ref{prop:estimate1} yields the desired estimate \eqref{estimate-ZB} by noting the definition of $\widetilde\varphi$. 
		\end{proof}
	Since the sequence $\{T_k\}$ is arbitrarily chosen and the solution of \eqref{B-inf} is unique, we can obtain that the solution of the BSDE \eqref{B-T} converges to the solution of the BSDE in \eqref{B-inf} as $T\rightarrow \infty.$ 
		%%%%%%%%%%%%%%%%%%%%%%%%%%%%%%%
		%%%%%%%%%%%%%%%%%%%%%%%%%%%%%%%
		\subsubsection{The limit of $C^{(T)}$}\label{sec:convergence-C}
		From the analysis above, we have the convergence from $(\overline A^{(T)},Z^{\overline A,(T)})$ and $(\overline B^{(T)},Z^{\overline B,(T)})$ to $(\overline A,Z^{\overline A})$ and $(\overline B,Z^{\overline B})$, respectively, as $T\rightarrow \infty.$ Now we are ready to show the convergence result of $C^{(T)}$.
		
		Let
		$$
		h^{(T)}_t:=e^{-\phi t}\sigma^2_t\frac{2\overline A^{(T)}_{t}-\gamma}{2\gamma^2}-\frac{\rho^2_t(\overline B^{(T)}_{t})^2}{4(\gamma\rho_t+\frac{\phi}{2}\gamma+\lambda_t)} + e^{-\frac{\phi t}{2}}\sigma_t\frac{Z^{\overline B,(T)}_{t}}{\gamma},
		$$
		and
		$$
		h_t:=e^{-\phi t}\sigma^2_t\frac{2\overline A_{t}-\gamma}{2\gamma^2}-\frac{\rho^2_t(\overline B_{t})^2}{4(\gamma\rho_t+\frac{\phi}{2}\gamma+\lambda_t)}+e^{-\frac{\phi t}{2}}\sigma_t\frac{Z^{\overline B}_{t}}{\gamma}. 
		$$
		From the BSDE \eqref{C-T}, we trivially extend $C^{(T)}$ from $[0,T]$ to $[0,\infty)$ as follows 
		\begin{equation*}
			\left\{\begin{split}
				C^{(T)}_t=&~\mathbb E\left[ \left. \int^T_t h^{(T)}_s\,ds\right|\mathcal F_t\right],\quad &t\in[0,T),\\
				C^{(T)}_t=&~0,\quad &t\in[T,\infty).
			\end{split}\right.
		\end{equation*}
Thus, $Z^{C,(T)}$ is also trivially extended to $[0,\infty)$.
  
		Define $C$ as
		\begin{equation}\label{def:C-martingale}
			C_t=\mathbb E\left[ \left.  \int^\infty_t h_s\,ds\right|\mathcal F_t\right].
		\end{equation}
		The next theorem shows the wellposedness of infinite horizon BSDE for $C$ and the convergence from $C^{(T)}$ to $C$.
		\begin{theorem}\label{thm:C}
			The infinite horizon BSDE \eqref{C-inf} admits a unique solution $(C,Z^{C})$ in $S^{2,0}(0,\infty)\times \mathcal M^{2,0}(0,\infty)$. Moreover, $( C^{(T)},Z^{C,(T)}     )$ converges to $(C,Z^C)$ in $S^{2,0}(0,\infty)\times \mathcal M^{2,0}(0,\infty)$.
		\end{theorem} 
		\begin{proof}
			First, the conditional expectation \eqref{def:C-martingale} is well defined. Indeed, due to the boundedness of $\overline A^{}$ and Theorem \ref{thm:B}, it holds that 
		\begin{equation*}\label{estimate-h}
				\begin{split}
					\mathbb E\left[\left(\int^\infty_0 \left|h_s\right|\,ds\right)^2\right]
					\leq &~c\left(\mathbb E\left[\left(\int^\infty_0 e^{-\phi s}\,ds\right)^2\right]+\mathbb E\left[\sup_{t\geq 0}e^{\frac{\phi t}{2}}\left|\overline B^{}_{t}\right|^4\right]+\mathbb E\left[\int^\infty_0e^{\frac{\phi t}{2}}\left(Z^{\overline B}_{s}\right)^2\,ds\right]\right)\\
					<&~\infty,
				\end{split}	
			\end{equation*}
			and that
			\[
			\mathbb E\left[  \sup_{t\geq 0} C^2_t     \right]\leq c\mathbb E\left[ \left( \int_0^\infty |h_s|\,ds\right)^2\right] <\infty. 
			\]
			Thus, $M_\cdot:=C_\cdot+\int_0^\cdot h_s\,ds$ is square integrable martingale. The martingale representation theorem implies the existence of a stochastic process $Z^C\in\mathcal M^{2,0}_{}(0,\infty)$ such that 
			$$M_t=M_0+\int^t_0 Z^C_s\,dW_s.$$
			Thus, for any $0\leq t\leq u<\infty,$ we have that
			\begin{equation*}
				C_u-C_t=-\int^u_t h_s\,ds+\int^u_t Z^C_s \,dW_s,
			\end{equation*}
			from which we can conclude that $(C,Z^C)$ is a solution to the BSDE \eqref{C-inf}.

			By the definition of $C^{(T)}$ and $C$, we have that
			\begin{equation*}
				\begin{split}
					%		&~\left| C^{(T)}_t-C_t\right|^2 \leq 2\left(\mathbb E\left[\left. \int^T_0 \left| h^{(T)}_s-h_s\right| \,ds\right|\mathcal F_t\right]\right)^2+ 2\left(	\mathbb E\left[\left.	\int_T^\infty h_s\,ds \right|\mathcal F_t   \right] \right)^2 \\
					&~\mathbb E\left[\sup_{0\leq t\leq T}\left| C^{(T)}_t-C_t\right|^2\right] \leq 2\mathbb E\left[\sup_{0\leq t\leq T} \left(\mathbb E\left[ \left. \int^T_0 \left| h^{(T)}_s-h_s\right| \,ds\right|\mathcal F_t\right]\right)^2 \right]+2\mathbb E\left[\sup_{0\leq t\leq T}\left(	\mathbb E\left[\left.	\int_T^\infty h_s\,ds \right|\mathcal F_t   \right] \right)^2\right]\\
					\leq&~c \mathbb E\left[  \left( \int^T_0 \left| h^{(T)}_s-h_s\right| \,ds \right)^2 \right]+c\mathbb E\left[ \left( \int_T^\infty |h_s|\,ds \right)^2  \right]\qquad\qquad  (\textrm{by Doob's maximal inequality})\\
					\leq&~ c\mathbb E\left[  \left(  \int_0^\infty e^{-\phi t}|	\overline A_t^{(T)}  - \overline A_t	|\,dt  \right)^2  \right] +c \mathbb E\left[ \left( \int_0^\infty   \left| \left(\overline B^{(T)}_t\right)^2-\overline B_t^2  \right|\,dt      \right)^2    \right] \\
					&~+ c\mathbb E\left[   \left(  \int_0^\infty  e^{ -\frac{\phi}{2}t }  |  Z^{\overline B,(T)}_t - Z^{\overline B}_t			|\,dt     \right)^2     \right]		+c\mathbb E\left[ \left( \int_T^\infty |h_s|\,ds \right)^2  \right]		\\
					\leq&~c\mathbb E\left[    \sup_{t\geq 0}  e^{-\frac{\phi t}{4}} \left|	\overline A^{(T)}_t-\overline A_t	\right|^2  \left( \int_0^\infty e^{-\frac{7\phi t}{8}}\,dt\right)^2 \right] + c \mathbb E\left[  \sup_{t\geq 0} \left(  e^{ \frac{\phi t}{4}  }\left|\overline B^{(T)}_t-\overline B_t\right|    \right)^2 \left(     \int_0^\infty e^{-\frac{\phi t }{4}  } \left|\overline B^{(T)}_t+\overline B_t\right|\,dt  \right)^2  \right]  \\
					&~+ c\mathbb E\left[    \int_0^\infty  e^{\frac{\phi}{2}t} \left|	Z^{\overline B,(T)}_t - Z^{\overline B}_t	\right|^2\,dt \left(\int_0^\infty e^{-\frac{3\phi}{2}t}\,dt    \right)^2   \right]    +c\mathbb E\left[ \left( \int_T^\infty |h_s|\,ds \right)^2  \right] \qquad (\textrm{by H\"older's inequality})\\
					\leq&~c\mathbb E\left[   \sup_{t\geq 0} e^{-\phi t} \left|	\overline A^{(T)}_t-\overline A_t	\right|^8 \right]^{1/4} + c \mathbb E\left[  \sup_{t\geq 0}  e^{ \phi t  }\left|\overline B^{(T)}_t-\overline B_t\right|^4\right]^{1/2} \mathbb E\left[\left(     \int_0^\infty e^{-\frac{\phi t}{4}   } \left|\overline B^{(T)}_t+\overline B_t\right|\,dt  \right)^4  \right]^{ 1/2 }  \\
					&~+ c\mathbb E\left[   \left( \int_0^\infty  e^{\frac{\phi t}{2}} \left|	Z^{\overline B,(T)}_t - Z^{\overline B}_t	\right|^2\,dt  \right)^2 \right]^{1/2}  +c\mathbb E\left[ \left( \int_T^\infty |h_s|\,ds \right)^2  \right]  \qquad (\textrm{using again H\"older's inequality})\\
					\leq&~ c\mathbb E\left[   \sup_{t\geq 0} e^{-\phi t} \left|	\overline A^{(T)}_t-\overline A_t	\right|^8 \right]^{1/4}+ c \mathbb E\left[  \sup_{t\geq 0}  e^{ \phi t}\left|\overline B^{(T)}_t-\overline B_t\right|^4\right]^{1/2} \mathbb E\left[	\sup_{t\geq 0} e^{\phi t} \left(\left| \overline B^{(T)}_t\right|^4+\left|\overline B_t\right|^4\right)			\left(     \int_0^\infty e^{-\frac{\phi t}{2}  }\,dt    \right)^4  \right]^{ 1/2 }  \\
					&~+ c\mathbb E\left[   \left( \int_0^\infty  e^{\frac{\phi t}{2}} \left|	Z^{\overline B,(T)}_t - Z^{\overline B}_t	\right|^2\,dt  \right)^2 \right]^{1/2} +c\mathbb E\left[  \left(\int_T^\infty |h_s|\,ds\right)^2    \right]  \\
					\leq&~ c\mathbb E\left[   \sup_{t\geq 0} e^{-\phi t} \left|	\overline A^{(T)}_t-\overline A_t	\right|^8 \right]^{1/4}  +   c \mathbb E\left[  \sup_{t\geq 0}  e^{\phi t }\left|\overline B^{(T)}_t-\overline B_t\right|^4\right]^{1/2} \\
					&~ + c\mathbb E\left[   \left( \int_0^\infty  e^{\frac{\phi t}{2}} \left|	Z^{\overline B,(T)}_t - Z^{\overline B}_t	\right|^2\,dt  \right)^2 \right]^{1/2}  +c\mathbb E\left[ \left( \int_T^\infty |h_s|\,ds \right)^2  \right]  \\
					\xrightarrow{T\rightarrow\infty}&~ 0 ,
				\end{split}
			\end{equation*}
			where the last convergence is due to Theorem \ref{thm:A} and Theorem \ref{thm:B}. 
   
   From the BSDE representation of $C^{(T)}$ and $C$, we have
			\begin{equation*}
				\begin{split}
					\mathbb E\left[  \int_0^\infty \left| Z^{C,(T)}_r - Z^C_r		\right|^2\,dr     \right]\leq c\mathbb E\left[	\sup_{t\geq 0}\left|	C^{(T)}_t- C_t		\right|^2	\right] + c\mathbb E\left[ \left(\int_0^\infty  |			h^{(T)}_r-h_r|\,dr \right)^2     \right]\xrightarrow{T\rightarrow\infty} 0.
				\end{split}
			\end{equation*}
		\end{proof}

		%%%%%%%%%%%%%%%%%%%%%%%%%
		%%%%%%%%%%%%%%%%%%%%%%%%%
		\subsection{Convergence of control problems}
		
		With the convergence result in Section \ref{sec:convergence-ABC}, we have the convergence of the finite horizon stochastic control \eqref{cost-T-continuous}-\eqref{state-T-continuous}. In particular, we get the convergence of the value fucntion and the optimal position. 
		Moreover, we will verify that the limit can characterize the value function as well as the optimal strategy of the infinite horizon stochastic control problem \eqref{cost-inf-continuous}-\eqref{state-inf-continuous}. 
		
		\begin{theorem}\label{thm-inf}
			Let the {\bf Standing Assumption} be satisfied. Let $(A,B,C)$ be the limit in Section \ref{sec:convergence-ABC}.

			\begin{itemize}
				\item[i)] Define 
				\begin{equation} \label{value-inf}
					V(t,\mathcal X):=\mathcal X^\top A_t\mathcal X+\mathcal X^\top B_t+C_t. 
				\end{equation}
				Then $V(t,\mathcal X)$ is the value function of \eqref{cost-inf-continuous}-\eqref{state-inf-continuous} starting at $t$ with initial state $\mathcal X$.
				\item[ii)] Define $\Delta\widetilde X^*$ as  
				\begin{equation}\label{jump-inf}
					\Delta  \widetilde{X}^*_{t}= \frac{I^{A}_t}{a_t} {\mathcal X}_{t-}+\frac{I^{B}_t}{a_t},
				\end{equation}
				where the processes $I^A,I^B$ are given by 
				\begin{equation*}
					I^{  A}=\left(\begin{array}{c}
						-\rho \overline A_{}-\lambda \\
						- \frac{ \rho }{\gamma}\overline A_{}+\rho +\frac{\phi}{2}
					\end{array}\right)^{\top},
					\qquad
					I^{  B}=- \frac{\rho }{2}\overline B_{}.
				\end{equation*}
				Moreover, for each $s\geq t$ define 
				\begin{equation}\label{state-inf}
					\widetilde X^*_s=\left(1-\frac{\lambda_s+\rho_s \overline A_{ s}}{a_s}\right) \widetilde P^*_s-\frac{\rho_s \overline B_{s}}{2a_s} \quad \text{and} \quad   \widetilde Y^*_s=\frac{\gamma(\lambda_s+\rho_s \overline A_{s})}{a_s} \widetilde P^*_s+\frac{\gamma\rho_s \overline B_{s}}{2a_s},  
				\end{equation}
				where $\widetilde P^*$ satisfies the dynamics:
				\[
				d\widetilde P^*_s=-\left(\frac{\phi}{2}+\frac{\rho_s(\lambda_s+\rho_s \overline A_{s})}{a_s}\right)\widetilde P^*_s\,ds-\frac{ \rho^2_s \overline B_{s}}{2a_s}\,ds + \frac{\sigma_s}{\gamma} e^{  -\frac{1}{2}\phi s  }  \,dW_s,\quad \widetilde P^*_0=x_0.
				\]
				If $\widetilde X^*$ is a c\`adl\`ag semimartingale, then $( X^*_\cdot, Y^*_\cdot):=(e^{\frac{\phi}{2}\cdot}\widetilde X^*_\cdot, e^{\frac{\phi}{2}\cdot} \widetilde Y^*_\cdot)$ is the optimal state of the stochastic control problem with infinite horizon \eqref{cost-inf-continuous}-\eqref{state-inf-continuous}. 
				\item[iii)] Recall the constant $L$ in \eqref{ass:lambda}. If we further assume $\frac{\phi}{2}<L<\phi$, then we have 
					\begin{equation}\label{convergence-X}
						\lim_{T\rightarrow\infty}\mathbb E\left[ \sup_{0\leq t\leq T}e^{-\phi t}|X^{*,(T)}_t- X^*_t|^2 \right] =0.
					\end{equation}
			\end{itemize}
		\end{theorem}		
		\begin{proof}%[Proof of Theorem \ref{thm-inf}]
			{\bf Step 1.}
			{\it For any admissible strategy $\widetilde X$, we will verify
			\begin{equation}\label{estimate-Y}
			\mathbb E\left[  \sup_{t\geq 0} e^{ Lt } | \widetilde Y_t |^4    \right]+\mathbb E\left[ \left(\int_0^\infty  e^{\frac{1}{2}Ls}\widetilde Y^2_s\,ds \right)^2  \right]<\infty,
			\end{equation}
			where we recall $L\in(0,\phi)$ is the constant appearing in the definition of the space of admissible strategies.
			}
			Let $\widetilde P=\widetilde X+\frac{1}{\gamma}\widetilde Y$. It holds that 
				\begin{equation*}
					\begin{split}
						d\widetilde P_s%=&~d\widetilde X_s+\frac{1}{\gamma}d\widetilde Y_s\\
						%		=&~ d\widetilde X_s+			\frac{1}{\gamma}\left\{	\left(-\frac{\phi}{2}\widetilde{Y}_s-\rho \widetilde{Y}_s-\frac{\phi\gamma}{2}\widetilde{X}_s\right)\,ds-\gamma\,d \widetilde{X}_t+\sigma_s e^{-\frac{\phi s}{2}}\,dW_s \right\}\\
						%		=&~	\frac{1}{\gamma} 	\left(-\frac{\phi}{2}\gamma P_s-\rho \widetilde{Y}_s \right)\,ds +\frac{1}{\gamma}\sigma_s e^{-\frac{\phi s}{2}}\,dW_s  \\
						%		=&~-\frac{\phi}{2} P_s- \rho_s ( P_s-\widetilde X_s)  \,ds   +	\frac{1}{\gamma}\sigma_s e^{-\frac{\phi s}{2}}\,dW_s  \\ 
						=  \left\{  -\left(    \frac{\phi}{2} + \rho_s     \right)\widetilde P_s + \rho_s\widetilde X_s \right\} \,ds+	\frac{1}{\gamma}\sigma_s e^{-\frac{\phi s}{2}}\,dW_s,\quad \widetilde P_0=x_0. 
					\end{split}
			\end{equation*}     		
			It\^o's formula implies that
			\begin{equation*}
				\begin{split}
					&~e^{\frac{1}{2}Lt}\widetilde P^2_t+\int_0^t\left(  \phi+2\rho_s-\frac{L}{2}       \right) e^{\frac{1}{2}Ls}\widetilde P^2_s\,ds    \\%=&~x^4_0+\int_0^t    K e^{Ks} \widetilde P^4_s\,ds + 4\int_0^t e^{Ks} P^3_s\,dP_s+ 6\int_0^t e^{Ks} \frac{1}{\gamma^2}\sigma^2_s e^{-\phi s} \widetilde P^2_s\,ds\\
					%	=&~x^4_0+\int_0^t    K e^{Ks} \widetilde P^4_s\,ds + 4\int_0^t e^{Ks} \widetilde P^3_s\left\{      \left\{  -\left(    \frac{\phi}{2} + \rho_s     \right)P_s + \rho_s\widetilde X_s \right\} \,ds+	\frac{1}{\gamma}\sigma_s e^{-\frac{\phi s}{2}}\,dW_s     \right\}+ 6\int_0^t e^{Ks} \frac{1}{\gamma^2}\sigma^2_s e^{-\phi s} P^2_s\,ds\\
					=&~x_0^2+ 2\int_0^t \rho_s e^{\frac{1}{2}Ls}\widetilde P_s\widetilde X_s\,ds + \frac{2}{\gamma}\int_0^t\sigma_se^{\frac{1}{2}Ls-\frac{1}{2}\phi s} \widetilde P_s\,dW_s+ \frac{1}{\gamma^2}\int_0^t e^{\frac{1}{2}Ls} \sigma^2_s e^{-\phi s} \,ds,
				\end{split}
			\end{equation*}
			which further implies by Young's inequality 
			\begin{equation}\label{verification-estimate-Y1}
				\begin{split}
					&~e^{\frac{1}{2}Lt}\widetilde P^2_t+\int_0^t\left(  \phi+2\rho_s-\frac{L}{2}       \right) e^{\frac{L}{2}s}\widetilde P^2_s\,ds \\
					\leq &~x_0^2+ \delta\int_0^t  e^{\frac{1}{2}Ls}\widetilde P^2_s \,ds + \frac{c}{\delta} \int_0^t e^{\frac{1}{2}Ls}\widetilde X^2_s\,ds + \frac{2}{\gamma}\int_0^t\sigma_se^{\frac{1}{2}Ls-\frac{1}{2}\phi s} \widetilde P_s\,dW_s+ c.
				\end{split}
			\end{equation}
			Taking $\mathbb E\left[ \sup_{t\geq 0} \cdots    \right]$ on both sides of squared  \eqref{verification-estimate-Y1} and using BDG's inequality, we have 
			\begin{equation*}
				\begin{split}
					&~	\mathbb E\left[ \sup_{t\geq 0} e^{Lt} \widetilde P^4_t     \right]+ \mathbb E\left[ \left(\int_0^\infty\left(  \phi+2\rho_s-\frac{L}{2}   -\delta     \right) e^{\frac{1}{2}Ls}\widetilde P^2_s\,ds \right)^2  \right] \\
					\leq&~  c\mathbb E\left[ \left( \int_0^\infty e^{\frac{1}{2}Ls}\widetilde X^2_s\,ds \right)^2  \right] + c\mathbb E\left[ \int_0^\infty e^{ (L-\phi)s   } \widetilde P^2_s\,ds  \right] +c\\
					\leq&~c\mathbb E\left[ \left( \int_0^\infty e^{\frac{1}{2}Ls}\widetilde X^2_s\,ds \right)^2  \right] + c\mathbb E\left[ \sup_{s\geq 0}e^{\frac{1}{2}Ls}\widetilde P^2_s \int_0^\infty e^{ (\frac{1}{2}L-\phi)s   }  \,ds  \right] +c\\
					\leq&~c\mathbb E\left[ \left( \int_0^\infty e^{\frac{1}{2}Ls}\widetilde X^2_s\,ds \right)^2  \right] + c\mathbb E\left[ \sup_{s\geq 0}e^{\frac{1}{2}Ls}\widetilde P^2_s   \right] +c\\
     \leq&~       c\mathbb E\left[ \left( \int_0^\infty e^{\frac{1}{2}Ls}\widetilde X^2_s\,ds \right)^2  \right] + \frac{1}{4}\mathbb E\left[ \sup_{s\geq 0}e^{Ls}\widetilde P^4_s   \right] +c  .     
				\end{split}
			\end{equation*}
			Thus,
			\[
			\mathbb E\left[ \sup_{t\geq 0} e^{Lt} \widetilde P^4_t     \right]+ \mathbb E\left[ \left(\int_0^\infty  e^{\frac{1}{2}Ls}\widetilde P^2_s\,ds \right)^2  \right] \leq c\mathbb E\left[ \left( \int_0^\infty e^{\frac{1}{2}Ls}\widetilde X^2_s\,ds \right)^2  \right]  +c<\infty,
			\]
			which implies \eqref{estimate-Y} by the definition of $\widetilde P$ and the admissible conditions for $\widetilde X$.

			{\bf Step 2.} {\it In this step, for any admissible strategy $\widetilde X$
			we will verify the convergence of the cost functional (up to a subsequence) as $T\rightarrow\infty$,
			%\[
			%	\mathbb E\left[ \int_0^T \frac{1}{a_s} \left(  I_{s}^{\overline A,(T)}\mathcal X^{}_{s}+  I^{\overline B,(T)}_s\right)^2\,ds   \right] \rightarrow 	\mathbb E\left[ \int_0^\infty \frac{1}{a_s} \left(  I_{s}^{\overline A }\mathcal X^{ }_{s}+  I^{\overline B}_s\right)^2\,ds   \right]. 
			%\]
    which together with \eqref{cost-T-2} in Appendix \ref{sec:verification-T} will imply that
			\[
			J(t,\widetilde X) = \mathbb E\left[\left.	\int_t^\infty \frac{1}{a_s} \left(I^{A}_s\mathcal X_s+I^{B}_s\right)^2\,ds 	\right|\mathcal F_t	\right] + \mathcal X^\top_{t-}A_t\mathcal X_{t-} + \mathcal X^\top_{t-}B_t+C_t.
			\]
		}
			First, by the uniform boundedness of $\overline A^{(T)}$, we have  
			\begin{equation*}
				\begin{split}
					&~	  \mathbb E\left[ \left| \mathbb E\left[ \left.  \int_t^\infty   \frac{1}{a_s} \left(  I_{s}^{A,(T)}\mathcal X^{}_{s}+  I^{B,(T)}_s\right)^2\,ds  \right|  \mathcal F_t   \right] -  \mathbb E\left[ \left. \int_t^T   \frac{1}{a_s} \left(  I_{s}^{A,(T)}\mathcal X^{}_{s}+  I^{B,(T)}_s\right)^2\,ds  \right|\mathcal F_t   \right] \right| \right] \\
					\leq &~c\mathbb E\left[   \int_T^\infty  	\frac{(c+\lambda_s)^2}{a_s}\widetilde X^2_s+ \widetilde Y^2_s+ |\overline B^{(T)}_s		|^2\,ds	    \right]\\
					\leq&~ c\mathbb E\left[    \int_T^\infty  	(1+\lambda_s)\widetilde X^2_s+ \widetilde Y^2_s+ |\overline B^{(T)}_s		|^2\,ds	    \right].
					%		\leq&~c\mathbb E\left[ \sup_{s\geq 0}e^{\frac{1}{2}Ks}\widetilde X^2_s \int_T^\infty (1+\lambda_s) e^{-\frac{1}{2}Ks}\,ds    \right]+c\mathbb E\left[ \sup_{s\geq 0} e^{\frac{1}{2}Ks} \widetilde Y^2_s \int_0^\infty e^{-\frac{1}{2}Ks} \,ds   \right] + c\mathbb E\left[ \sup_{s\geq 0} e^{\frac{1}{2}\phi s} \left|\overline B^{(T)}_s\right|^2 \int_0^\infty  e^{-\frac{1}{2}\phi s}\,ds \right]\\
					%		<&~\infty.
				\end{split}
			\end{equation*}
			Note that by {\bf Step 1}, the admissible conditions for $\widetilde X$, the assumption for $\lambda$ and Theorem \ref{thm:B}, it holds that 
			\begin{equation*}
				\begin{split}
					&~	\mathbb E\left[ \int_0^\infty  	(1+\lambda_s)\widetilde X^2_s+ \widetilde Y^2_s+ |\overline B^{(T)}_s		|^2\,ds	   \right]\\
					\leq&~c\mathbb E\left[ \sup_{s\geq 0}e^{\frac{1}{2}Ls}\widetilde X^2_s \int_0^\infty (1+\lambda_s) e^{-\frac{1}{2}Ls}\,ds     \right]+c\mathbb E\left[ \sup_{s\geq 0} e^{\frac{1}{2}Ls} \widetilde Y^2_s \int_0^\infty e^{-\frac{1}{2}Ls} \,ds   \right]\\
					&~ + c\mathbb E\left[ \sup_{s\geq 0} e^{\frac{\phi s}{2}}\left|\overline B^{(T)}_s\right|^2 \int_0^\infty  e^{-\frac{\phi s}{2}}\,ds \right]\\
					<&~\infty.
				\end{split}
			\end{equation*}
			Thus, 
			\begin{equation*}
				\begin{split}
				\mathbb E\left[	\left| \mathbb E\left[ \left. \int_t^\infty   \frac{1}{a_s} \left(  I_{s}^{A,(T)}\mathcal X^{}_{s}+  I^{B,(T)}_s\right)^2\,ds  \right|\mathcal F_t   \right] -  \mathbb E\left[ \left. \int_t^T   \frac{1}{a_s} \left(  I_{s}^{A,(T)}\mathcal X^{}_{s}+  I^{B,(T)}_s\right)^2\,ds   \right|\mathcal F_t  \right]	\right| \right]\xrightarrow{T\rightarrow\infty} 0.
				\end{split}
			\end{equation*}
			
			Second, by the boundedness of $\overline A^{(T)}$ and $\overline A$, and by noting that $\frac{\lambda}{a}$ is also bounded, we have 
			\begin{equation*}
				\begin{split}
					&~\mathbb E\left[   \left|\mathbb E\left[ \left. \int_t^\infty   \frac{1}{a_s} \left(  I_{s}^{A,(T)}\mathcal X^{}_{s}+  I^{B,(T)}_s\right)^2\,ds  \right|\mathcal F_t   \right]  -  \mathbb E\left[\left. \int_t^\infty \frac{1}{a_s} \left(  I_{s}^{A}\mathcal X^{ }_{s}+  I^{B}_s\right)^2\,ds \right|\mathcal F_t  \right]\right|  \right]  \\
					\leq&~\mathbb E\left[   \int_0^\infty \frac{1}{a_s} \left|    (	I^{A,(T)}_s-I^{A}_s		)\mathcal X_s+I^{B,(T)}_s-I^{B}_s 	        \right|	\left|    (	I^{A,(T)}_s + I^{A}_s		)\mathcal X_s+I^{B,(T)}_s+I^{B}_s 	        \right|   \,ds 		\right]\\
					\leq&~ c \mathbb E\left[     \int_0^\infty \left\{  | \overline A^{(T)}_s-\overline A_s ||\widetilde X_s| + | \overline A^{(T)}_s-\overline A_s ||\widetilde Y_s|  + |\overline B^{(T)}_s - \overline B_s|  \right\} \left\{      |\widetilde X_s| + |\widetilde Y_s| + |\overline B^{(T)}_s| + | \overline B_s |     \right\}   \,ds   \right]\\
					\leq&~ c \mathbb E\left[     \int_0^\infty \left\{  | \overline A^{(T)}_s-\overline A_s ||\widetilde X_s|   \right\} \left\{      |\widetilde X_s| + |\widetilde Y_s| + |\overline B^{(T)}_s| + | \overline B_s |     \right\}   \,ds    \right]\\
					&~ +c \mathbb E\left[     \int_0^\infty \left\{  | \overline A^{(T)}_s-\overline A_s ||\widetilde Y_s|  \right\} \left\{      |\widetilde X_s| + |\widetilde Y_s| + |\overline B^{(T)}_s| + | \overline B_s |     \right\}   \,ds     \right]\\
					&~ +c \mathbb E\left[   \int_0^\infty \left\{ |\overline B^{(T)}_s - \overline B_s|  \right\} \left\{      |\widetilde X_s| + |\widetilde Y_s| + |\overline B^{(T)}_s| + | \overline B_s |     \right\}   \,ds     \right]\\
					:=&~ I_1+I_2+I_3.
				\end{split}
			\end{equation*}
			For $I_1$, we have %by Theorem \ref{thm:A}, Theorem \ref{thm:B}, the admissible conditions for $\widetilde X$ and {\bf Step 1}
			\begin{equation*}
				\begin{split}
					I_1\leq&~ c \mathbb E\left[    \int_0^\infty   | \overline A^{(T)}_s-\overline A_s |       \left\{      |\widetilde X_s|^2 + |\widetilde Y_s|^2 + |\overline B^{(T)}_s|^2 + | \overline B_s |^2     \right\}   \,ds       \right] \\
					\leq&~ 	c\mathbb E\left[   \int_0^\infty 	 | \overline A^{(T)}_s-\overline A_s |  \left(   |\widetilde X_s|^2+  |\widetilde Y_s|^2\right)\,ds	\right]	+c \mathbb E\left[ 	\sup_{s\geq 0} e^{-\frac{1}{8}\phi s}| \overline A^{(T)}_s-\overline A_s | \sup_{s\geq 0} e^{\frac{1}{2}\phi s}(  |\overline B^{(T)}_s|^2 + | \overline B_s |^2  )\int_0^\infty e^{-\frac{3}{8}\phi s}\,ds	\right]\\
					\leq&~ c\mathbb E\left[   \int_0^\infty 	 | \overline A^{(T)}_s-\overline A_s |  \left(   |\widetilde X_s|^2+  |\widetilde Y_s|^2\right)\,ds	\right]  + c\mathbb E\left[  \sup_{s\geq 0} e^{-\frac{1}{4}\phi s}| \overline A^{(T)}_s-\overline A_s |^2  \right]^{1/2}\mathbb E\left[   \sup_{s\geq 0} e^{\phi s} (  |\overline B^{(T)}_s|^4 + | \overline B_s |^4  ) \right]  \\
			%		&\leq~ c\mathbb E\left[\left.  \sup_{s\geq 0} e^{ -\frac{1}{8}\phi s } |	\overline A^{(T)}_s - \overline A_s	| \left\{  \sup_{s\geq 0} e^{\frac{1}{4}\phi s} |\widetilde X_s|^2 + 	  \sup_{s\geq 0} e^{\frac{1}{4}\phi s}  |\widetilde Y_s|^2	+	 \sup_{s\geq 0} e^{\frac{1}{4}\phi s} |\overline B^{(T)}_s|^2 +  \sup_{s\geq 0} e^{\frac{1}{4}\phi s} |\overline B_s|^2   \right\}   \right|\mathcal F_t  \right]\\
				%	&\leq~ c\mathbb E\left[  \left. \sup_{s\geq 0} e^{ -\frac{\phi s}{4}}  |	\overline A^{(T)}_s - \overline A_s	|^2  \right|\mathcal F_t  \right]^{ \frac{1}{2} } \left\{ \mathbb E\left[	\left.	\sup_{s\geq 0} e^{ \frac{1}{2}\phi s }|	\widetilde X_s	|^4  \right|\mathcal F_t \right]^{\frac{1}{2}}	+\mathbb E\left[	\left.	\sup_{s\geq 0} e^{ \frac{1}{2}\phi s }|	\widetilde Y_s	|^4  \right|\mathcal F_t \right]^{\frac{1}{2}} \right.\\
				%	&~\left.\quad\qquad \qquad\qquad\qquad \qquad\qquad \qquad \quad	+ \mathbb E\left[	\left.	\sup_{s\geq 0} e^{ \phi s }|	\overline B^{(T)}_s	|^4  \right|\mathcal F_t \right]^{\frac{1}{2}}	 + \mathbb E\left[	\left.	\sup_{s\geq 0} e^{ \phi s }|	\overline B_s	|^4  \right|\mathcal F_t \right]^{\frac{1}{2}}		\right\}\\
					\overset{T\rightarrow\infty}{\longrightarrow}&~0\qquad (\textrm{up to a subsequence}),
				\end{split}
			\end{equation*}
		where the convergence of the second term is due to Theorem \ref{thm:A} and Theorem \ref{thm:B}, and to get the convergence of the first term, we note that Theorem \ref{thm:A} yields a subsequence $\overline A^{(T_k)}$ such that for each $s\geq 0$ $|\overline A^{T_k}-\overline A|\rightarrow 0$ a.s., and that $\int_0^\infty | \overline A^{(T_k)}_s-\overline A_s |  \left(   |\widetilde X_s|^2+  |\widetilde Y_s|^2\right)\,ds \leq c\int_0^\infty e^{\frac{1}{2}L s} \left(   |\widetilde X_s|^2+  |\widetilde Y_s|^2\right)\,ds $ that is integrable by {\bf Step 1}. Thus, dominated convergence implies the convergence of the first term.
		
			Similarly, $I_2+I_3\rightarrow 0$ as $T\rightarrow\infty$.

			Third, for the subsequence $\{T_k\}$, by Theorem \ref{thm:A}, Theorem \ref{thm:B}, Theorem \ref{thm:C}, there exists a further subsequence such that $  |A_t^{(T_{n_k})}-A_t| +  |B_t^{(T_{n_k})}-B_t|  +  |C^{(T_{n_k})}_t-C_t| \rightarrow 0$ a.s..

			Thus, for any $t\geq 0$ and for any admissible $\widetilde X$, $J^{(T)}(t,\widetilde X)\rightarrow J(t,\widetilde X)$ up to a subseuqence. In particular, it holds that $J^{(T)}(0,\widetilde X)\rightarrow J(0,\widetilde X)$.
			
			{\bf Step 3.} {\it We verify that $\widetilde X^{*}$ in \eqref{state-inf} is admissible.} 
			To do so, we first estimate $\widetilde P^*$. 
			
		%	Fix $0<K<\phi$. 
			It\^o's formula implies that 
			\begin{equation*}
				\begin{split}
					&~e^{Lt} (\widetilde P^*_t)^4 +\int_0^t \left(  	2\phi+\frac{4\rho_s( \lambda_s+\rho_s\overline A_s   )}{a_s} - L   \right) e^{Ls} (\widetilde P_s^*)^4\,ds \\
					=&~ x_0^4-2\int_0^t e^{Ls}(\widetilde P^*_s)^3\frac{\rho^2_s\overline B_s}{a_s}\,ds + 4\int_0^t e^{Ls-\frac{1}{2}\phi s} (\widetilde P^*_s)^3\frac{\sigma_s}{\gamma}\,dW_s+6\int_0^t e^{Ls}(\widetilde P^*_s)^2\frac{\sigma^2_s}{\gamma^2}e^{-\phi s}\,ds ,
				\end{split}
			\end{equation*}
			By BDG's inequality and Young's inequality, we have
			\begin{equation*}
				\begin{split}
					&~	\mathbb E\left[	\sup_{t\geq 0}e^{Lt} (\widetilde P^*_t)^4\right]  +  \mathbb E\left[\int_0^\infty \left(  	2\phi+\frac{4\rho_s( \lambda_s+\rho_s\overline A_s   )}{a_s} - L	   \right) e^{Ls} (\widetilde P_s^*)^4\,ds\right]\\
					\leq&~x_0^4+ c\mathbb E\left[\int_0^\infty e^{Ls}(\widetilde P^*_s)^3 \overline B_s \,ds  \right] + c\mathbb E\left[ \left(  \int_0^\infty    		e^{2Ls-\phi s}(\widetilde P_s^*)^6		\,ds         \right)^{\frac{1}{2}}  \right]+ c\mathbb E\left[     \sup_{s\geq 0}  e^{\frac{1}{2}Ls} (\widetilde P^*_s)^2		\int_0^t e^{  \frac{1}{2}Ls-\phi s } 	 \,ds            \right]\\
					\leq&~ \delta\mathbb E\left[   \sup_{t\geq 0} e^{Lt} (\widetilde P^*_t)^4    \right] + c \mathbb E\left[  \left(  \int_0^\infty  e^{ \frac{1}{4}Ls } | \overline B_s	|   \,ds     \right)^4    \right]	+  \delta\mathbb E\left[  \sup_{t\geq 0} e^{Lt} (\widetilde P^*_t)^4        \right] + c,
				\end{split}
			\end{equation*}
			which implies that by letting $\delta$ be small enough
			\begin{equation*}
				\begin{split}
					\mathbb E\left[	\sup_{t\geq 0}e^{Lt} (\widetilde P^*_t)^4\right]  \leq  c  \mathbb E\left[  \sup_{s\geq 0}  e^{ \phi s } | \overline B_s	|^4      \right]	  + c <\infty. 
				\end{split}
			\end{equation*}
			By \eqref{state-inf}, %for $0<K<\phi$
			\begin{equation*}
				\begin{split}
					\mathbb E\left[ \sup_{s\geq 0} e^{Ls} |\widetilde X^*_s|^4      \right] \leq&~ c\mathbb E\left[ \sup_{s\geq 0} e^{Ls} \frac{(\gamma\rho_s+\frac{\phi \gamma}{2}-\rho_s\overline A_s)^4}{a_s^4} (\widetilde P^*_s)^4    \right] + c\mathbb E\left[  \sup_{s\geq 0} e^{Ls} |\overline B_s|^4         \right]\\
					\leq&~c\mathbb E\left[ \sup_{s\geq 0} e^{Ls}  (\widetilde P^*_s)^4    \right] + c\mathbb E\left[  \sup_{s\geq 0} e^{Ls} |\overline B_s|^4         \right] <\infty. 
				\end{split}
			\end{equation*}
			Moreover, applying It\^o's formula to $e^{\frac{1}{2}Lt}(\widetilde P^*_t)^2$ we get
			\begin{equation*}
				\begin{split}
					&~e^{ \frac{1}{2}Lt }(\widetilde P^*_t)^2  +\int_0^t \left( 	\phi+\frac{ 2\rho_s(   \lambda_s+\rho_s\overline A_s  )  }{a_s}-\frac{L}{2}	    \right) e^{ \frac{1}{2}Ls  }(\widetilde P^*_s)^2 \,ds  \\
					=&~ x_0^2 - \int_0^t e^{ \frac{1}{2}Ls  } \widetilde P^*_s \frac{	\rho^2_s\overline B_s	}{a_s}\,ds+2\int_0^t e^{ \frac{1}{2}Ls -\frac{1}{2}\phi s  }\widetilde P^*_s \frac{\sigma_s}{\gamma}\,dW_s+\int_0^t \frac{\sigma^2_s}{\gamma} e^{	\frac{1}{2}Ls-\phi s	}\,ds\\
					\leq&~ \delta\int_0^te^{\frac{1}{2}Ls} (\widetilde P^*_s)^2\,ds + c \int_0^t e^{ \frac{1}{2}Ls }\overline B^2_s\,ds +2\int_0^t e^{ \frac{1}{2}Ls -\frac{1}{2}\phi s  }\widetilde P^*_s \frac{\sigma_s}{\gamma}\,dW_s + c,
				\end{split}
			\end{equation*}
			which implies that 
			\begin{equation*}
				\begin{split}
					\mathbb E\left[  \left( \int_0^\infty e^{\frac{1}{2}Ls} (\widetilde P^*_s)^2 \,ds  \right)^2   \right] \leq&~ c+c\mathbb E\left[  \left( \int_0^\infty e^{\frac{1}{2}Ls} \overline B^2_s \,ds  \right)^2  \right]+ c\mathbb E\left[ 	\int_0^\infty e^{Ls-\phi s}(\widetilde P^*_s)^2\,ds		  \right]\\
					\leq&~c+ c\mathbb E\left[\sup_{s\geq 0}e^{\phi s}\overline B^4_s \left(\int_0^\infty   e^{\frac{1}{2}(L-\phi)s}     \,ds	\right)^2	\right]+c\mathbb E\left[  \sup_{s\geq 0}e^{\frac{1}{2}Ls}(\widetilde P^*_s)^2	\int_0^\infty  e^{\frac{1}{2}Ls-\phi s}   \,ds	   \right]\\
					\leq&~ c+ c\mathbb E\left[\sup_{s\geq 0}e^{\phi s}\overline B^4_s \right] + c\mathbb E\left[  \sup_{s\geq 0}e^{\frac{1}{2}Ls}(\widetilde P^*_s)^2\right]\\
					<&~\infty. 
				\end{split}
			\end{equation*}
			Thus, 
			\begin{equation*}
				\begin{split}
					\mathbb E\left[   \left( \int_0^\infty e^{\frac{1}{2}Ls}(\widetilde X^*_s)^2\,ds    \right)^2   \right]\leq&~ c\mathbb E\left[  \left( \int_0^\infty e^{\frac{1}{2}Ls} (\widetilde P^*_s)^2 \,ds  \right)^2   \right]+ c\mathbb E\left[  \left( \int_0^\infty e^{\frac{1}{2}Ls} \overline B_s^2 \,ds  \right)^2   \right]\\
					\leq&~c\mathbb E\left[  \left( \int_0^\infty e^{\frac{1}{2}Ls} (\widetilde P^*_s)^2 \,ds  \right)^2   \right]+ c\mathbb E\left[\sup_{s\geq 0}e^{\phi s}\overline B^4_s \right]\\
					<&~\infty.
				\end{split}
			\end{equation*}

			{\bf Step 4.} {\it Complete the verification.} By Proposition \ref{thm:finite-control} (refer to \eqref{cost-T-2} in Appendix \ref{sec:verification-T}), it holds that
			\begin{equation*}
				\begin{split}
					J^{(T)}(t,\widetilde X) =&~\mathbb E\left[\left.	\int_t^T \frac{1}{a_s} \left(I^{ A,(T)}_s\mathcal X_s+I^{B,(T)}_s\right)^2\,ds 	\right|\mathcal F_t	\right] + \mathcal X^\top_{t-}A^{(T)}_t\mathcal X_{t-} + \mathcal X^\top_{t-}B^{(T)}_t+C^{(T)}_t \\
					\geq&~ \mathcal X^\top_{t-}A^{(T)}_t\mathcal X_{t-} + \mathcal X^\top_{t-}B^{(T)}_t+C^{(T)}_t.
				\end{split}
			\end{equation*}
			Taking limit on both sides, by {\bf Step 2} and {\bf Step 3}, we have 
			\begin{equation*}
				\begin{split}
					J(t, \widetilde X) \geq \mathcal X^\top_{t-}A_t\mathcal X_{t-} + \mathcal X^\top_{t-}B_t+C_t.
				\end{split}
			\end{equation*}
			It is easy to check that $I^{A}\mathcal X^{*}+I^{B}=0$ for $\mathcal X:=(\widetilde X^*,\widetilde Y^*)$ defined in \eqref{state-inf}. Thus, $\widetilde X^*$ in \eqref{jump-inf} and \eqref{state-inf} is optimal if it is a c\`adl\`ag semimartingale.

			{\bf Step 5.} {\it Convergence of optimal positions under the assumption $\frac{\phi}{2}<L<\phi$.} To verify the convergence \eqref{convergence-X}, we first consider the convergence from $\widetilde P^{(T)}$ from $\widetilde P^*$, where $\widetilde P^{(T)}$ is defined in Proposition \ref{thm:finite-control}. By definition, the difference follows 
			\begin{equation*}
				\left\{\begin{split}
					d(\widetilde P^{(T)}_t - \widetilde P^*_t) =&~ \left\{- \frac{\rho^2_t}{a_t}( \overline A^{(T)}_t - \overline A_t	)\widetilde P^*_t - \frac{\rho^2_t}{2a_t}(	\overline B^{(T)}_t - \overline B_t		)\right\}\,dt  - \left( \frac{\phi}{2} + \frac{	\rho_t(  \lambda_t+\rho_t\overline A^{(T)}_t  )		}{a_t}   \right)(  \widetilde P^{(T)}_t - \widetilde P^*_t      )   \,dt\\
					\widetilde P^{(T)}_0 - \widetilde P^*_0 = &~ 0,
				\end{split}\right.
			\end{equation*}
		which implies that 
		\[
			\widetilde P^{(T)}_t - \widetilde P^*_t = \int_0^t 	\left\{- \frac{\rho^2_s}{a_s}( \overline A^{(T)}_s - \overline A_s	)\widetilde P^*_s - \frac{\rho^2_s}{2a_s}(	\overline B^{(T)}_s - \overline B_s		)\right\}	e^{	-\int_s^t \left( \frac{\phi}{2} + \frac{	\rho_r(  \lambda_r+\rho_r\overline A^{(T)}_r  )		}{a_r}   \right) \,dr	}	\,ds.
		\]
		Thus, we have
		\begin{equation*}
			\begin{split}
				&~\mathbb E\left[  \sup_{0\leq t\leq T}   \left|\widetilde P^{(T)}_t - \widetilde P^*_t \right|^2		   \right] \\
				\leq&~ c\mathbb E\left[  \left(\int_0^\infty   | \overline A^{(T)}_s - \overline A_s	| |\widetilde P^*_s|     \,ds  \right)^2  \right] + c \mathbb E\left[  \left(\int_0^\infty   |\overline B^{(T)}_s-\overline B_s		|\,ds \right)^2 \right] \\
				\leq&~ c\mathbb E\left[  \sup_{s\geq 0} e^{-\frac{\phi s}{4}} |	\overline A^{(T)}_s-\overline A_s	|^2	\sup_{s\geq 0} e^{\frac{1}{2}Ls}|\widetilde P^*_s|^2	\left(\int_0^\infty e^{\frac{1}{8}\phi s-\frac{1}{4}Ls} \,ds  \right)^2      \right] \\
				&~+ \mathbb E\left[  \sup_{s\geq 0} e^{\frac{1}{2}\phi s}| \overline B^{(T)}_s - \overline B_s |^2 \left( \int_0^\infty e^{-\frac{1}{4}\phi s} \,ds \right)^2   \right]\\
				\overset{T\rightarrow\infty}{\longrightarrow} &~0.
			\end{split}
		\end{equation*}
		By \eqref{state-T-X} and \eqref{state-inf}, we have 
		\begin{equation*}
			\begin{split}
				&~\mathbb E\left[\sup_{0\leq t\leq T}\left| \widetilde X^{*,(T)}_t - \widetilde X^*_t		\right|^2\right] \\
				\leq&~ c \mathbb E\left[\sup_{0\leq t\leq T}\left| \widetilde P^{(T)}_t - \widetilde P^*_t		\right|^2\right] + c\mathbb E\left[\sup_{0\leq t\leq T}e^{-\frac{\phi t}{4}}\left| \overline A^{(T)}_t - \overline A_t		\right|^2	e^{\frac{Lt}{2}}| \widetilde P^*_t	|^2	e^{\frac{\phi t}{4}-\frac{Lt}{2}}	\right]+c\mathbb E\left[\sup_{0\leq t\leq T}\left| \overline B^{*,(T)}_t - \overline B^*_t		\right|^2\right]\\
				\leq&~ c \mathbb E\left[\sup_{0\leq t\leq T}\left| \widetilde P^{(T)}_t - \widetilde P^*_t		\right|^2\right]  + c\mathbb E\left[\sup_{0\leq t\leq T}e^{-\frac{\phi t}{2}}\left| \overline A^{(T)}_t - \overline A_t		\right|^4	 \right]^{1/2} \mathbb E\left[\sup_{0\leq t\leq T}e^{Lt}| \widetilde P^*_t	|^4	\right]^{1/2}\\
				&~ +c\mathbb E\left[\sup_{0\leq t\leq T}e^{\frac{\phi t}{2}}\left| \overline B^{*,(T)}_t - \overline B^*_t		\right|^2\right]\\
				\overset{T\rightarrow\infty}{\longrightarrow}&~ 0,
			\end{split}
		\end{equation*}
	which yields \eqref{convergence-X}.
		\end{proof}

	\section{Properties of the optimal solution}\label{sec:property}	
In this section, we explore three properties of the optimal solution under the following assumption, in addition to {\bf Standing Assumption} in Section \ref{sec:intro}.

{\bf Additional Assumption.} $\rho$ and $\lambda$ are constants.

Under {\bf Standing Assumption} and {\bf Additional Assumption}, it holds that $\overline A^{(T)}$ is deterministic, $\overline A$ is a constant, and $\overline B^{(T)}=\overline B\equiv 0$. 

{\bf Property 1: {\it the trading behavior of the investor in the long run}.} By Theorem \ref{thm-inf} and $\overline B=0$, the optimal position $X$ on $[0,\infty)$ follows 
\begin{equation}\label{eq:X-property1}
    dX_t= \  -\frac{   \rho(\rho\overline A+\lambda )    }{\gamma\rho+\lambda+\frac{\phi\gamma}{2}}   X_t\,dt + \gamma\sigma_t\left(1- \frac{\rho\overline A+\lambda}{\gamma\rho+\lambda+\frac{\phi\gamma}{2}}  \right) \,dW_t.  
\end{equation}
The unique solution of \eqref{eq:X-property1} has the following closed form expression
\begin{equation}\label{eq:X-property1-closed}
X_t= \left(1- \frac{\rho\overline A+\lambda}{\gamma\rho+\lambda+\frac{\phi\gamma}{2}}  \right)X_0 e^{- \rho \frac{\rho\overline A+\lambda}{\gamma\rho+\lambda+\frac{\phi\gamma}{2}}  t  } + \left(1- \frac{\rho\overline A+\lambda}{\gamma\rho+\lambda+\frac{\phi\gamma}{2}}  \right) e^{ -\rho  \frac{\rho\overline A+\lambda}{\gamma\rho+\lambda+\frac{\phi\gamma}{2}}  t } \int_0^t \frac{\sigma_s }{\gamma} e^{ \rho  \frac{\rho\overline A+\lambda}{\gamma\rho+\lambda+\frac{\phi\gamma}{2}} s  } \,dW_s.
\end{equation}

We consider the following two cases.

{\bf Case 1: $\sigma$ is damped, i.e., $\sigma_t=\underline \sigma e^{-\psi t}$ with $\psi>0$ large enough.} In this case, setting $\overline\psi:=\psi-\frac{\rho(\rho\overline A+\lambda)}{\gamma\rho+\lambda+\frac{\phi\gamma}{2}}>0$, we have by Dambis-Dubins-Schwartz theorem 
\begin{equation}
	\int_0^te^{-\overline\psi s}\,dW_s = \overline W_{\tau(t)},
\end{equation} 
where $\overline W$ is a Brownian motion and $\tau(t)=\frac{1}{-2\overline\psi}(  e^{-2\overline\psi t}  -1    )$ that is bounded. Thus, 
$$
\lim_{t\rightarrow\infty} e^{ -\rho  \frac{\rho\overline A+\lambda}{\gamma\rho+\lambda+\frac{\phi\gamma}{2}}  t } \int_0^t e^{-\overline\psi s}\,dW_s=0\quad a.s.,  
$$
which implies that $\lim_{t\rightarrow\infty}X_t=0$ a.s. from \eqref{eq:X-property1-closed}. Thus, the investor will liquidate her asset finally.

{\bf Case 2: $\mathbb E[\sigma^2]:=\overline\sigma\neq 0$ is a constant.} 
\eqref{eq:X-property1-closed} implies that 
%\begin{equation}\label{estimate:X2-property1}
%	\begin{split}
%		\mathbb E[|X_t|^2] 
%		=&~ \left(1- \frac{\rho\overline A+\lambda}{\gamma\rho+\lambda+\frac{\phi\gamma}{2}}  \right)^2X^2_0 e^{- 2\rho \frac{\rho\overline A+\lambda}{\gamma\rho+\lambda+\frac{\phi\gamma}{2}}  t  } \\
%		&~ +   \left(1- \frac{\rho\overline A+\lambda}{\gamma\rho+\lambda+\frac{\phi\gamma}{2}}  \right)^2 e^{ -2 \rho  \frac{\rho\overline A+\lambda}{\gamma\rho+\lambda+\frac{\phi\gamma}{2}}  t } \int_0^t \frac{\mathbb E[\sigma^2_s] }{\gamma^2 } e^{ 2 \rho  \frac{\rho\overline A+\lambda}{\gamma\rho+\lambda+\frac{\phi\gamma}{2}} s  } \,ds 
%	\end{split}
%\end{equation}
%
%It holds that 
for $t\geq 0$ large enough
\begin{align*}
   % \begin{split}
        \mathbb E[|X_t|^2] 
     %   =&~ \left(1- \frac{\rho\overline A+\lambda}{\gamma\rho+\lambda+\frac{\phi\gamma}{2}}  \right)^2X^2_0 e^{- 2\rho \frac{\rho\overline A+\lambda}{\gamma\rho+\lambda+\frac{\phi\gamma}{2}}  t  }  +   \left(1- \frac{\rho\overline A+\lambda}{\gamma\rho+\lambda+\frac{\phi\gamma}{2}}  \right)^2 e^{ -2 \rho  \frac{\rho\overline A+\lambda}{\gamma\rho+\lambda+\frac{\phi\gamma}{2}}  t } \int_0^t \frac{\sigma^2 }{\gamma^2 } e^{ 2 \rho  \frac{\rho\overline A+\lambda}{\gamma\rho+\lambda+\frac{\phi\gamma}{2}} s  } \,ds  \\
        =&~\left(1- \frac{\rho\overline A+\lambda}{\gamma\rho+\lambda+\frac{\phi\gamma}{2}}  \right)^2X^2_0 e^{- 2\rho \frac{\rho\overline A+\lambda}{\gamma\rho+\lambda+\frac{\phi\gamma}{2}}  t  } \nonumber\\
        &~ +  \frac{\overline\sigma}{\gamma^2 } \left(1- \frac{\rho\overline A+\lambda}{\gamma\rho+\lambda+\frac{\phi\gamma}{2}}  \right)^2 e^{ -2 \rho  \frac{\rho\overline A+\lambda}{\gamma\rho+\lambda+\frac{\phi\gamma}{2}}  t } \frac{1}{2 \rho  \frac{\rho\overline A+\lambda}{\gamma\rho+\lambda+\frac{\phi\gamma}{2}}} \left(  e^{ 2 \rho  \frac{\rho\overline A+\lambda}{\gamma\rho+\lambda+\frac{\phi\gamma}{2}} t  } - 1 \right)  \nonumber\\
        \geq &~ \frac{\overline\sigma }{\gamma^2 } \left(1- \frac{\rho\overline A+\lambda}{\gamma\rho+\lambda+\frac{\phi\gamma}{2}}  \right)^2   \frac{1}{4 \rho  \frac{\rho\overline A+\lambda}{\gamma\rho+\lambda+\frac{\phi\gamma}{2}}}  \nonumber\\
        >&~0 \qquad (\textrm{since }A\textrm{ is valued in }[0,\gamma/2]) \nonumber.
%    \end{split}
\end{align*}
From \eqref{eq:X-property1} we can get immediately that $\mathbb E[\sup_{t\geq 0}|X_t|^2]<\infty$. 
 If $\lim_{t\rightarrow\infty} |X_t|^2=0$ a.s., dominated convergence theorem implies that $\lim_{\rightarrow\infty}\mathbb E[|X_t|^2]=0$, which is a contradiction. In this case, the investor never liquidates, but her position fluctuates around zero; the trading sign of the investor will be fully determined by the trading sign of the external flow, as shown in {\bf Property 2}.

{\bf Property 2: {\it the investor will buy (sell) when external sell (buy) orders arrive.}}
{In the second property, we want to investigate the trading direction at the next trading time, when the current position reaches zero. To explore this property, we return to the discrete time model by assuming trading occurs at $k\Delta$, $k = 0, 1, \cdots$. W.l.o.g., we only consider the case $k>0$, where the trading is continuous. By \eqref{eq:X-property1} we have the dynamics in  discrete time 
\begin{equation*} 
    X_{(n+1)\Delta}=X_{n\Delta}-\rho\frac{\rho\overline{A}+\lambda}{\gamma\rho+\lambda+\frac{\phi \gamma}{2}}X_{n\Delta}\Delta +  \gamma\frac{\gamma\rho+\frac{\phi \gamma}{2}-\rho\overline{A}}{\gamma\rho+\lambda+\frac{\phi \gamma}{2}}\sigma_{n\Delta} \epsilon_{n+1}.
\end{equation*}
When the current position reaches zero, i.e., $X_{n\Delta}=0$, we have 
\[
X_{(n+1)\Delta}=\gamma\frac{\gamma\rho+\frac{\phi \gamma}{2}-\rho\overline{A}}{\gamma\rho+\lambda+\frac{\phi \gamma}{2}} \sigma_{n\Delta} \epsilon_{n+1}
\]
Since $\overline{A}$ is $[0,\frac{\gamma}{2}]$-valued by Theorem \ref{thm:A}, we have $\frac{\gamma\rho+\frac{\phi \gamma}{2}-\rho\overline{A}}{\gamma\rho+\lambda+\frac{\phi \gamma}{2}}>0$.
When $\sigma_{n\Delta}\epsilon_{n+1}>0$ (sell orders), it holds that $X_{(n+1)\Delta}-X_{n\Delta}=X_{(n+1)\Delta}>0$, that is, the investor will buy. The arrival of external buy orders can be analyzed in a similar manner.
% Since 
% $$Y_{(n+1)\Delta-} = (1-\Delta\rho_{n\Delta})( Y_{n\Delta-} + \gamma\xi_n   ) + \sigma_{n\Delta} \epsilon_{n+1}$$
% When $\epsilon>0$, $Y_{n+1}$ will increase, which means other investors will sell the portfolios (performance like $\xi_n>0$).
% The execution price $\widetilde S_{n+1}=S_{n+1}-Y_{n+1}$ will \textcolor{red}{decrease} when $Y_{n+1}$ increase, where $S_t$ is the fundamental price, and this result is  with the result that the investor in our model will buy some portfolios. 
% similarly, when $\epsilon<0$, which means other investors will \textcolor{red}{buy} the portfolios, the execution price will increase, and the investor in our model will \textcolor{red}{sell} the portfolios.
} 

{\bf Property 3: {\it a turnpike property of the optimal solution.}} Our goal is to show that the stochastic control problem (1.1) %\eqref{cost-T-continuous} 
has the following turnpike property as in \cite{SWJ-2022}:

\textit{ 
there exist positive constants $K$ and $\mu $ independent of $T$ such that 
\begin{equation}\label{turnpike}
    |\mathbb{E}[X^{*,(T)}_t]|+|\mathbb{E}[Y^{*,(T)}_t]|\leq K[e^{-\mu t}+e^{-\mu(T-t)}], \quad \forall t\in [0,T].
\end{equation}
} 
% \begin{theorem}
%     We can find $(0,0)$ that there exist  positive constants $K, \mu > 0$, independent of $T$
% \begin{equation}\label{Turnpike property}
%     |\mathbb{E}[X^{*,(T)}_t]-0|+|\mathbb{E}[Y^{*,(T)}]-0|\leq K[e^{-\mu t}+e^{-\mu(T-t)}], \quad \forall t\in [0,T].
% \end{equation} 
% \end{theorem}
\begin{proof}[Proof of \eqref{turnpike}]
Let $P^{(T)}_\cdot=e^{\frac{1}{2}\phi \cdot}\widetilde P^{(T)}_\cdot$, where $\widetilde P^{(T)}$ was defined in Proposition \ref{thm:finite-control}.  Then
\begin{equation}\label{O}
    \begin{aligned}
        dP^{(T)}_t %&=dX^ {*, (T)}_t+\frac{1}{\gamma}dY^ {*, (T)}_t\\
   %     &=dX^ {*, (T)}_t-\frac{\rho}{\gamma} Y^ {*, (T)}_t~dt-dX^ {*, (T)}_t+\frac{\sigma_t}{\gamma}~dW_t\\
 %       &=-\frac{\rho}{\gamma} Y^ {*, %(T)}_t~dt+\frac{\sigma_t}{\gamma}~dW_t\\
        &=-\rho \frac{\lambda+\rho\overline A^{(T)}_t}{\gamma\rho+\lambda+\frac{\phi\gamma}{2}} P^{(T)}_t~dt+\frac{\sigma_t}{\gamma}\,dW_t, 
    \end{aligned}
\end{equation}
which implies that
\[ 
d\mathbb E[P^{(T)}_t]=-\rho \frac{\lambda+\rho\overline A^{(T)}_t}{\gamma\rho+\lambda+\frac{\phi\gamma}{2}}  \mathbb E[P^{(T)}_t]\,dt\Rightarrow \left|\mathbb E[P^{(T)}_t] \right|= |x_0| \exp\left( - \int_0^t \rho \frac{\lambda+\rho\overline A^{(T)}_s}{\gamma\rho+\lambda+\frac{\phi\gamma}{2}} \,ds \right)\leq |x_0| \exp\left( -\mu t   \right),
\]
where $\mu=\frac{\rho \lambda }{\gamma\rho+\lambda+\frac{\phi\gamma}{2}} $.

Since $\overline{A}^{(T)}$ is $[0,\frac{\gamma}{2}]$-valued, we can find $K_1>0$ and $K_2>0$ such that
\[ 
0<1-\frac{\lambda+\rho\overline A^{(T)}_t}{\gamma\rho+\lambda+\frac{\phi\gamma}{2}}=\frac{\gamma \rho+\frac{\phi\gamma}{2}-\rho \overline{A}^{(T)}_t}{\gamma \rho+\lambda+\frac{\phi\gamma}{2}}\leq \frac{\gamma\rho+\frac{\phi\gamma}{2}}{\gamma\rho+\lambda+\frac{\phi\gamma}{2}}  := K_1
\]
and  
\[ 
0<\gamma\frac{\lambda+\rho\overline A^{(T)}_t}{\gamma\rho+\lambda+\frac{\phi\gamma}{2}} \leq  \gamma\frac{\lambda+\frac{\rho\gamma}{2}}{\gamma\rho+\lambda+\frac{\phi\gamma}{2}}:=K_2
\] 
By \eqref{state-T-X} and \eqref{state-T-Y} in Proposition \ref{thm:finite-control}, it holds that 
\[
    X^{*,(T)}=\left( 1- \frac{\lambda+\rho\overline A^{(T)}}{  \gamma\rho+\lambda+\frac{\phi\gamma}{2}   } \right)P^{(T)},\qquad Y^{*,(T)} = \frac{ \gamma(\lambda+\rho\overline A^{(T)} )        }{\gamma\rho+\lambda+\frac{\phi\gamma}{2}}P^{(T)},
\]
which imply that 
\begin{equation*}
    \begin{split}
    \left|\mathbb{E}[X^{*,(T)}_t]\right|+\left|\mathbb{E}[Y^{*,(T)}_t]\right|=&~\left|\mathbb{E}\left[\left( 1- \frac{\lambda+\rho\overline A^{(T)}_t}{  \gamma\rho+\lambda+\frac{\phi\gamma}{2}   } \right)P^{(T)}_t\right]\right|+\left|\mathbb{E}\left[\frac{ \gamma(\lambda+\rho\overline A^{(T)}_t )        }{\gamma\rho+\lambda+\frac{\phi\gamma}{2}}P^{(T)}_t\right]\right|\\
    \leq&~K_1 |x_0| e^{-\mu t}+K_2 |x_0| e^{-\mu t}\\
    \leq&~ K[e^{-\mu t}+e^{-\mu(T-t)}],
    \end{split}
\end{equation*}
where $K= (K_1+K_2)|x_0|$. 
\end{proof}

\begin{center}
\begin{figure}[h]%\label{figure:position}
\includegraphics[width=3.33in, height=3in]{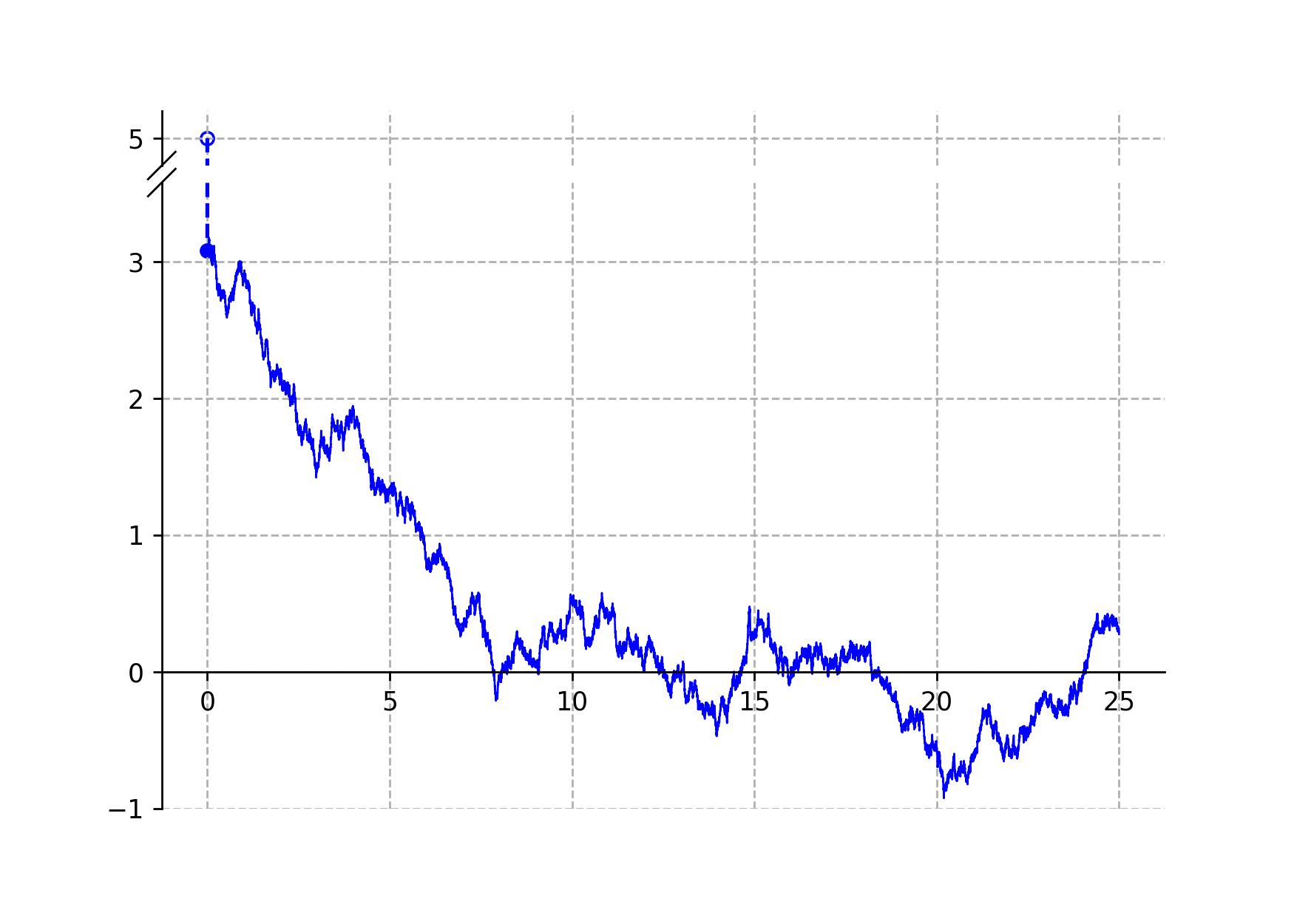}
\includegraphics[width=3.33in, height=3in]{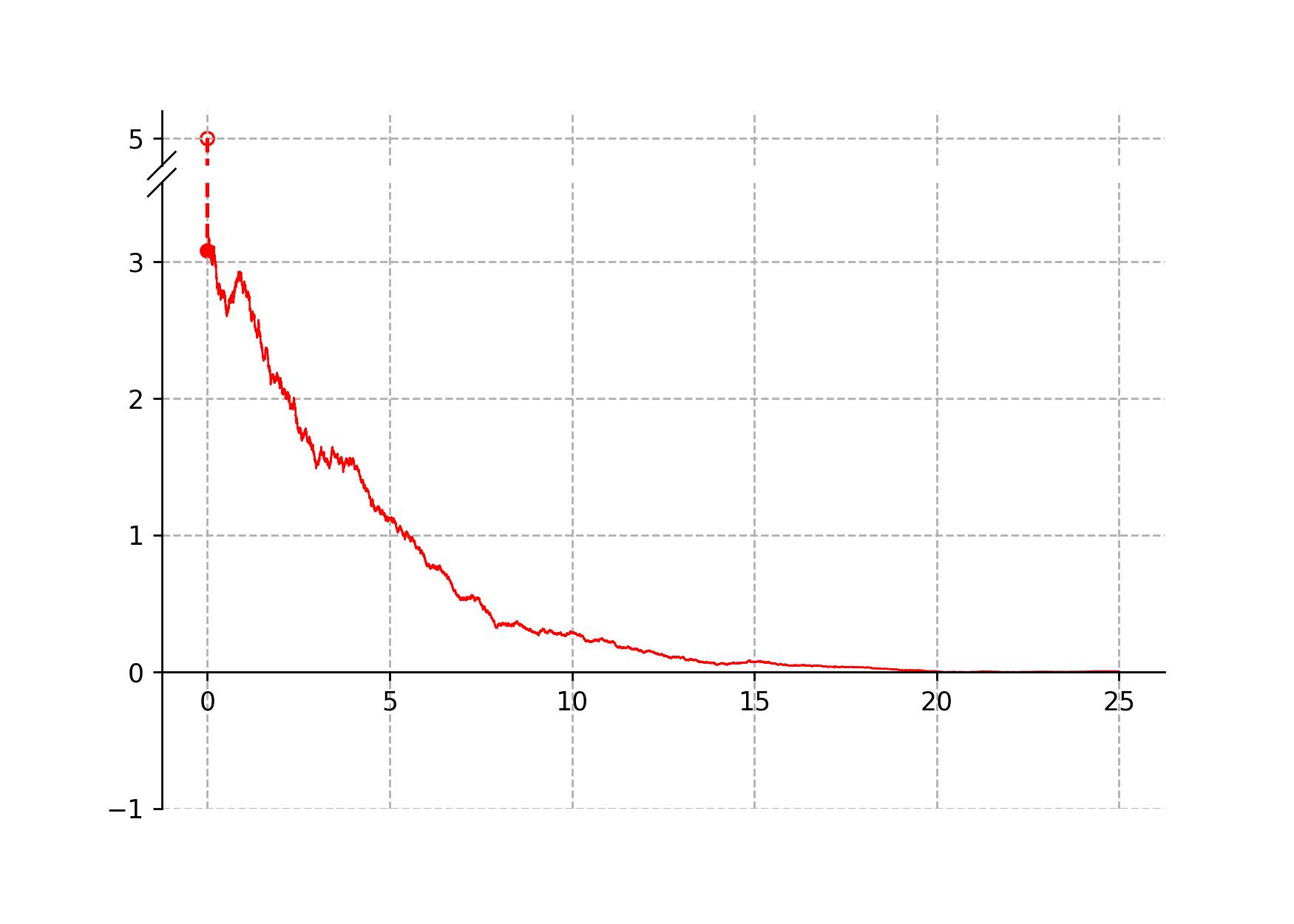}
\caption{Optimal position in the long run: $x_0=5$, $y_0=0$, $\rho=0.7$, $\gamma=0.5$, $\lambda=0.1$ and $\phi=0.2$. Left: $\sigma=0.25$;  Right: $\sigma_t=0.25 e^{-0.2t}$.}
\end{figure}
\end{center}
In Figure 1, we plot the optimal position as a function of time. The left panel exhibits the one with a constant external flow intensity $\sigma=0.25$ while the right panel exhibits the one with a damped external flow intensity $\sigma_t=0.25 e^{-0.2t}$. We see that the player never liquidate her assets but her position fluctuates around zero in the long run if the intensity for the external flow is a constant. In contrast, if the intensity is sufficiently damped, the investor will liquidate in the long term. This is consistent with {\bf Property 1}.  
%	This chapter shows how the optimal X value is affected by different parameter values. When $\rho=0.7, \gamma=o.5, \lambda=0.1, \phi=0.2$, we compare the preference of $X^*$ when $\sigma=0.25$ and $\sigma_t=0.25e^{-0.2t}$.
%%	\begin{figure}[H]
	%%		\centering 
	%%		\includegraphics[width=0.7\textwidth]{x=5 sig-const.png} 
	%%		\caption{$\sigma$ is a constant, $\sigma=0.25, \rho=0.7, \gamma=0.5, \lambda=0.1, \phi=0.2$} 
	%%		\label{Fig.1} 
	%%	\end{figure}
%	
%	We can observe from Figure 1 that $X^*$ does not approach 0, but fluctuates around 0, when time goes infinite, but it is obviously unrealistic, for $X$ won't float after it's equal to zero.
%	
%%	\begin{figure}[H]
	%%		\centering 
	%%		\includegraphics[width=0.7\textwidth]{x=5 sig-non-const.png} 
	%%		\caption{$\sigma$ is related with t, $\sigma_t=0.25e^{-0.2t}, \rho=0.7, \gamma=0.5, \lambda=0.1, \phi=0.2$} 
	%%		\label{Fig.2} 
	%%	\end{figure}
%	
%	From Figure 2, we can have when $\sigma$ can decay by time and approach zero as time goes to infinite, $X^*$ can decay to 0 as time goes by, which means we can finally sell out all of our portfolios.
%	

\section{Concluding remarks}
In this paper, we study the long time behavior of a stochastic control problem arising in optimal liquidation problems with semimartingale strategies and external flows. Specifically, we study the rigorous limit of the control problem when the horizon goes to infinity. Such a convergence result requires new result of BSDE convergence. As a byproduct of our analysis, we solve a stochastic control problem with infinite horizon and semimartingale strategies. In the constant setting, we further identify three key properties of the optimal position: 1) Whether the investor will liquidate her assets at all in the long run would rely on the intensity of the external flow; the investor will not liquidate in the long term if the external flow is strong enough. 2) The trading direction of the investor is fully determined by the external flow. 3) We also identify a turnpike property for the optimal solution.

\appendix

\section{Proof of Proposition \ref{thm:finite-control}}\label{sec:verification-T}

\begin{proof}	
	For any admissible strategy $\widetilde X$, we have
	\begin{equation}\label{eq:cost-rewritten-1}
		\begin{split}
			J^{(T)}(t,\widetilde X)=&~	\mathbb E\left[  \left. \int_{[t,T)} \left(-\widetilde Y_{s-}\,d\widetilde X_s+ \frac{\gamma}{2}d[\widetilde X]_s - \sigma_s e^{-\frac{\phi}{2}s}\,d[\widetilde X,W]_s\right)	+\int_t^T\left(\lambda_s \widetilde{X}^2_s-\frac{\phi}{2}\widetilde Y_{s}\widetilde{X}_s\right)\,ds	 \right|\mathcal F_t \right]\\
			&~+ \mathbb E\left[\left.-\left(	\widetilde{Y}_{T-}-\frac{\gamma_2}{2}\Delta \widetilde{X}_T	\right)\Delta \widetilde{X}_T \right|\mathcal F_t \right]\\
			=&~\mathbb E\left[ \left. \int_{[t,T)} \left(-\widetilde Y_{s-}\,d\widetilde X_s+ \frac{\gamma}{2}d[\widetilde X]_s - \sigma_s e^{-\frac{\phi}{2}s}\,d[\widetilde X,W]_s\right)	+\int_t^T\left(\lambda_s \widetilde{X}^2_s-\frac{\phi}{2}\widetilde Y_{s}\widetilde{X}_s\right)\,ds	 \right|\mathcal F_t \right]\\
			&~+\mathbb E\left[ \frac{\gamma}{2} \widetilde{X}^2_{T-} +\widetilde{X}_{T-}\widetilde{Y}_{T-}\Big|\mathcal F_t 	\right] \quad (\textrm{ since }\widetilde{X}_T=0)\\
			=&~\mathbb E\left[ \left. \int_{[t,T)} \left(-\widetilde Y_{s-}\,d\widetilde X_s+ \frac{\gamma}{2}d[\widetilde X]_s - \sigma_s e^{-\frac{\phi}{2}s}\,d[\widetilde X,W]_s\right)	+\int_t^T\left(\lambda_s \widetilde{X}^2_s-\frac{\phi}{2}\widetilde Y_{s}\widetilde{X}_s\right)\,ds	 \right|\mathcal F_t \right]\\
			&~+\mathbb E\left[  \left.  \mathcal X^\top_{T-} A^{(T)}_{T}\mathcal X_{T-} +\mathcal X^\top_{T-} B^{(T)}_T+C_T\right|\mathcal F_t \right].
		\end{split}
	\end{equation}
	We start with the term $\mathbb E\left[ \mathcal X^\top_{T-} A^{(T)}_{T}\mathcal X_{T-} \right]$. Using It\^o's formula in \cite[Theorem 36]{Protter-2005}, 
	\begin{equation*}
		\begin{aligned}
			\mathbb{E}\left[\mathcal X_{T-}^\top A^{(T)}_{T} \mathcal X_{T-}\Big|\mathcal F_t \right]=\mathbb{E} & {\left[\int_{t}^{T-} 2\left(A^{(T)}_{s} \mathcal X_{s-}\right)^{\top} d \mathcal X_{s}+\operatorname{Tr}\left(A^{(T)}_{s} d[\mathcal X, \mathcal X]_{s}^{c}\right)\right.} \\
			& \left.+\sum_{t \leq s<T}\left(\mathcal X_{s}^{\top} A^{(T)}_{s} \mathcal X_{s}-\mathcal X_{s-}^{\top} A^{(T)}_{s} \mathcal X_{s-}-2\left(A^{(T)}_{s} \mathcal X_{s-}\right)^{\top} \Delta \mathcal X_{s}\right)\right.\\
			&\left.+\mathcal X^\top_{t-} A^{(T)}_{t} \mathcal X_{t-}+\int_{t}^{T} \Big( \mathcal X_{s}^{\top} \dot{A}^{(T)}_{s} \mathcal X_{s} \,ds+2\mathcal{D}_s^\top Z^{A,(T)}_{s} \mathcal X_s\, ds\right.\\
			&~ +2\mathcal{K}^\top Z^{A,(T)}_{s} \mathcal X_s \,d[\widetilde{X}^c,W]_s+\mathcal X_{s}^{\top} Z^{A,(T)}_s \mathcal X_{s} \,dW_s\Big) \Bigg|\mathcal F_t \Bigg] ,
		\end{aligned}
	\end{equation*}
	where  $\dot{A}^{(T)}$ and $Z^{A,(T)}:=(Z^{A,(T)}_{ij})_{2\times2}$ denote the driver and the diffusion term in the BSDE for $A^{(T)}$, respectively. Since $A^{(T)}_{11}=\gamma A^{(T)}_{21}$, $A^{(T)}_{12}=\gamma A^{(T)}_{22}+\frac{1}{2}$, we have that $Z^{A,(T)}_{11}=\gamma Z^{A,(T)}_{21}, Z^{A,(T)}_{12}=\gamma Z^{A,(T)}_{22}$, thus $\mathcal{K}^\top Z^{A,(T)}\equiv 0.$ It is easy to prove that $\int^{\cdot}_{0} \mathcal X_{s}^{\top} Z^{A,(T)}_s \mathcal X_{s} \,dW_s$ is a uniformly integrable martingale and thus we can drop this term. 
	
	Note that
	\begin{equation*}
		\begin{aligned}
			&d[\mathcal X, \mathcal X]_{s}^{c}\\
			&=\left(\begin{array}{ccc}
				d\left[\widetilde{X}^{c}, \widetilde{X}^{c}\right]_{s} & e^{-\frac{\phi s}{2}}\sigma_{s} d\left[\widetilde{X}^{c}, W\right]_{s}-\gamma d\left[\widetilde{X}^{c}, \widetilde{X}^{c}\right]_{s}  \\
				e^{-\frac{\phi s}{2}}\sigma_{s} d\left[\widetilde{X}^{c}, W\right]_{s}-\gamma d\left[\widetilde{X}^{c}, \widetilde{X}^{c}\right]_{s} & \sigma_{s}^{2}e^{-\phi s} d s+\gamma^{2} d\left[\widetilde{X}^{c}, \widetilde{X}^{c}\right]_{s}-2 \gamma \sigma_{s}e^{-\frac{\phi s}{2}} d\left[\widetilde{X}^{c}, W\right]_{s} 
			\end{array}\right),
		\end{aligned}
	\end{equation*}
	which implies by the relationship between the entries of the matrix $A^{(T)}$ that
	\[
	\int_{t}^{T-}\textnormal{Tr}( A^{(T)}_s\,d[\mathcal X^c,\mathcal X^c]_s )\,ds=\int_t^T\left( -\frac{\gamma}{2}\,d[\widetilde{X}^c,\widetilde{X}^c]_s+\sigma_se^{-\frac{\phi s}{2}}\,d[\widetilde{X}^c,W]_s+e^{-\phi s}\sigma^2_s A^{(T)}_{22,s}\,ds\right).
	\]		
	Taking this back into the decomposition of $\mathbb{E}\left[\mathcal X_{T-} A^{(T)}_{T} \mathcal X_{T-}\Big|\mathcal F_t \right] $, we have
	
	\begin{equation}\label{v4}
		\begin{aligned}
			\mathbb{E}\left[\mathcal X_{T-} A^{(T)}_{T} \mathcal X_{T-}\Big|\mathcal F_t \right] 
			= &~ \mathbb{E}\left[  \left. \int_{t}^{T} 2\left(A^{(T)}_{s} \mathcal X_{s}\right)^{\top}\mathcal{H}_s \mathcal X_{s}\,ds+\int_{t}^{T-} 2\left(A^{(T)}_{s} \mathcal X_{s-}\right)^{\top} \mathcal{K}\, d \widetilde{X}_{s}\right|\mathcal F_t \right] \\
			&~+\mathbb{E}\left[ \left.  \int_{t}^{T}\left(-\frac{\gamma}{2} d\left[\widetilde{X}^{c}, \widetilde{X}^{c}\right]_{s}+\sigma_{s}e^{-\frac{\phi s}{2}} d\left[\widetilde{X}^{c}, W\right]_{s}+\sigma_{s}^{2} e^{-\phi s}A^{(T)}_{22, s}\,ds\right)\right|\mathcal F_t \right] \\
			&~+  \mathbb{E}\left[\sum_{t \leq s<T}\left(\mathcal X_{s}^{\top} A^{(T)}_{s} \mathcal X_{s}-\mathcal X_{s-}^{\top} A^{(T)}_{s} \mathcal X_{s-}-2\left(A^{(T)}_{s} \mathcal X_{s-}\right)^{\top} \mathcal{K} \Delta \widetilde{X}_{s}\right)\Big|\mathcal F_t \right] \\
			&~ +\mathbb{E}\left[ \left.  \int_{t}^{T}\left( \mathcal X_{s}^{\top} \dot{A}^{(T)}_{s} \mathcal X_{s}+2\mathcal{D}_s^\top Z^{A,(T)}_s \mathcal X_s \right)\,ds\right|\mathcal F_t \right]+\mathcal X^\top_{t-} A^{(T)}_{t} \mathcal X_{t-}.
		\end{aligned}
	\end{equation}
	Besides, integration by parts implies that
	\begin{equation}\label{v5}
		\begin{aligned}
			&~\mathbb E\left[\mathcal X_{T-}^\top B^{(T)}_T \Big|\mathcal F_t \right]\\
			=&\mathbb E\left[ \left. \int^T_t\left\{ \left(\mathcal X^\top_s\mathcal{H}^\top_sB^{(T)}_s+\dot B^{(T)}_s \mathcal X_s+\mathcal{D}_s^\top Z^{B,(T)}_s\right)\,ds+ \mathcal{K}^\top Z^{B,(T)}_{s}\,d[\widetilde{X}^c,W]_s\right\}+\int^{T-}_t\mathcal K^\top B^{(T)}_s\,d\widetilde{X}_s\right|\mathcal F_t \right]\\
			&+\mathbb E \left[\left.\sum_{t \leq s<T}\left(B^\top_s \mathcal X_s-B^\top_s \mathcal X_{s-}\right)-B^\top_s\Delta \mathcal X_s\right|\mathcal F_t \right]+ \mathcal X^\top_{t-}B^{(T)}_t\\
			=&\mathbb E\left[ \left.  \int^T_t\left\{ \mathcal X^\top_s\mathcal{H}^\top_sB^{(T)}_s+\dot B^{(T)}_s \mathcal X_s+\mathcal{D}_s^\top Z^{B,(T)}_s \right\}\,ds+\int^{T-}_t\mathcal K^\top B^{(T)}_s\,d\widetilde{X}_s\right|\mathcal F_t \right]+ \mathcal X^\top_{t-}B^{(T)}_t,
		\end{aligned}
	\end{equation}
	where $\dot B^{(T)}$ is the driver and $Z^{B,(T)}:=\begin{pmatrix}
		Z^{B,(T)}_{1} & 
		Z^{B,(T)}_{2}
	\end{pmatrix}^\top$ denotes the diffusion term in the BSDE for $B^{(T)}$. Since $B^{(T)}_{1}=\gamma B^{(T)}_{2}$, we also get that $\mathcal{K}^\top Z^{B,(T)}\equiv 0.$

	Next, we collect all the terms in \eqref{v4} and \eqref{v5} involving jumps:
	\begin{equation}\label{v9}
		\begin{aligned}
			& \mathbb{E}\left[ \left.  \int_{t}^{T-} 2\left(A^{(T)}_{s} \mathcal X_{s-}\right)^{\top} \mathcal{K}\,d \widetilde{X}_{s}+\sum_{t \leq s<T}\left(\mathcal X_{s}^{\top} A^{(T)}_{s} \mathcal X_{s}-\mathcal X_{s-}^{\top} A^{(T)}_{s} \mathcal X_{s-}-2\left(A^{(T)}_{s} \mathcal X_{s-}\right)^{\top} \mathcal{K} \Delta \widetilde{X}_{s}\right)\right|\mathcal F_t \right] \\
			&+\mathbb{E}\left[\left. \int^{T-}_t\mathcal K^\top B^{(T)}_s\,d\widetilde{X}_s\right|\mathcal F_t \right].
		\end{aligned}
	\end{equation}
	
	Since
%	\begin{equation}\label{v10}
	$	\mathcal{K}^{\top} A^{(T)}=\left(\begin{array}{lll}
			1 & -\gamma 
		\end{array}\right) A^{(T)}=\left(\begin{array}{lll}
			0 & \frac{1}{2} 
		\end{array}\right)$,
%	\end{equation}
	we have that	
	\begin{equation*}
		\begin{aligned}
			2\left(A^{(T)}_{s} \mathcal X_{s-}\right)^{\top} \mathcal{K}\,d \widetilde{X}_{s} & =2 \mathcal{K}^{\top} A^{(T)}_{s} \mathcal X_{s-}\,d\widetilde{X}_{s}  =\left(\begin{array}{lll}
				0 & 1 
			\end{array}\right) \mathcal X_{s-}\,d\widetilde{X}_{s}  =\widetilde{Y}_{s-}\,d\widetilde{X}_{s}.
		\end{aligned}
	\end{equation*}
	
	Moreover, the definition of $\mathcal{K}$ implies that
	\begin{equation*}
		\begin{aligned}
			\mathcal X_{s}^{\top} A^{(T)}_{s} \mathcal X_{s}-\mathcal X_{s-}^{\top} A^{(T)}_{s} \mathcal X_{s-}-2\left(A^{(T)}_{s} \mathcal X_{s-}\right)^{\top} \mathcal{K} \Delta \widetilde{X}_{s}&=\Delta \mathcal X_{s}^{\top} A^{(T)}_{s} \Delta \mathcal X_{s} \\
			&=\Delta \widetilde{X}_{s} \mathcal{K}^{\top} A^{(T)}_{s} \Delta \mathcal X_{s}\\
			&=\frac{1}{2} \Delta \widetilde{X}_{s} \Delta \widetilde{Y}_{s}\\
			&=-\frac{\gamma}{2}\left(\Delta \widetilde{X}_{s}\right)^{2},
		\end{aligned}
	\end{equation*}
	and that
	$\mathcal K^\top B^{(T)}_s =B^{(T)}_1-\gamma B^{(T)}_2=0.$
	As a result, the jump terms \eqref{v9} together with the term $\mathbb{E}\left[ \left.  \int_{t}^{T}\left(-\frac{\gamma}{2} d\left[\widetilde{X}^{c}, \widetilde{X}^{c}\right]_{s}+ e^{-\frac{\phi s}{2}}\sigma_{s} d\left[\widetilde{X}^{c}, W\right]_{s}\right)  \right|  \mathcal F_t  \right]$ in \eqref{v4} cancel with the first three terms in \eqref{eq:cost-rewritten-1}. Hence, taking \eqref{v4} and \eqref{v5} into \eqref{eq:cost-rewritten-1} yields that
	\begin{equation*}
		\begin{split}
			&~ J^{(T)}(t,\widetilde X) \\
			= &~ \mathbb{E}\left[\int_{t}^{T} 2\left(A^{(T)}_{s} \mathcal X_{s}\right)^{\top} \mathcal{H}_s\mathcal X_s\,ds+\int^T_t \mathcal X^\top_s\mathcal{H}^\top_sB^{(T)}_s\,ds+\int^T_t 2\mathcal{D}^\top_sZ^{A,(T)}_s \mathcal X_s\,ds+\int^T_t \mathcal{D}_s^\top Z^{B,(T)}_s\,ds \right.\\
			&~~~\left.\left.+\int_{t}^{T} \mathcal{D}_{s}^{\top} A^{(T)}_{s} \mathcal{D}_{s}\,ds+\int_{t}^{T} \mathcal X_{s}^{\top} \mathcal{Q}_s \mathcal X_{s}\,ds+\int^T_t\mathcal X_{s}^{\top}\dot B^{(T)}_s\,ds +\int_{t}^{T} \mathcal X_{s}^{\top} \dot{A}_{s}^{(T)} \mathcal X_{s}\,d s+\int_{t}^{T} \dot{C}^{(T)}_{s}\,ds\right|\mathcal F_t \right] \\
			&~~~ +\mathcal X^\top_{t-} A^{(T)}_t \mathcal X_{t-}+\mathcal X^\top_{t-} B^{(T)}_t +C^{(T)}_t\\
			= &~ \mathbb{E}\left[\int_{t}^{T}\mathcal X_{s}^{\top}\left(\mathcal{Q}_s+\mathcal{H}^{\top}_s A^{(T)}_{s}+A^{(T)}_{s} \mathcal{H}_s +\dot{A}^{(T)}_{s}\right)\mathcal X_{s}\,d s\right. \\
			&~~~+\int^T_t\mathcal X_s^\top\left(\mathcal{H}^\top_sB^{(T)}_s+\dot B^{(T)}_s+2Z^{A,(T)}_s\mathcal{D}_s\right)ds\\
			&~~~\left.\left.+\int_{t}^{T}\left(\mathcal{D}_{s}^\top A^{(T)}_{s} \mathcal{D}_{s}+\dot{C}^{(T)}_{s}+\mathcal{D}_s^\top Z^{B,(T)}_s \right) ds\right|\mathcal F_t\right]  +\mathcal X^\top_{t-} A^{(T)}_t\mathcal X_{t-} +\mathcal X^\top_{t-} B^{(T)}_t+C^{(T)}_t,
		\end{split}
	\end{equation*}
where $\dot C^{(T)}$ denotes the driver of the BSDE for $C^{(T)}$. 
	We rewritte the driver for $(A^{(T)},B^{(T)},C^{(T)})$ in the matrix form as:	
	\begin{equation*}
		\left\{\begin{aligned}
			&~\frac{\left(I^{A,(T)}\right)^{\top} I^{A,(T)}}{a} =\mathcal{Q}+\mathcal{H}^{\top} A^{(T)}+A^{(T)} \mathcal{H}+\dot{A}^{(T)}, \\
			&~\frac{2\left(I^{A,(T)}\right)^{\top}I^{B,(T)}}{a}=\mathcal{H}^\top B^{(T)}+2Z^{A,(T)}\mathcal{D}+\dot B^{(T)},\\
			&~\frac{(I^{B,(T)})^2}{a} =\mathcal{D}^{\top} A^{(T)}\mathcal D + \mathcal D^\top Z^{B,(T)} +\dot{C}^{(T)}.
		\end{aligned}\right.
	\end{equation*}
	Thus, we can get the alternative expression for the cost functional:
	\begin{equation}\label{cost-T-2}
		J^{(T)}(t, \widetilde X)=\mathbb{E}\left[  \left.  \int_{t}^{T} \frac{1}{a_s}(I_{s}^{A,(T)}\mathcal X_{s}+I^{B,(T)}_s)^{2}\,ds\right |\mathcal F_t\right]+\mathcal X^\top_{t-} A^{(T)}_t\mathcal X_{t-} +\mathcal X^\top_{t-} B^{(T)}_t+C^{(T)}_t.
	\end{equation}	
	On the one hand, we see that
	\begin{equation*}\label{alt}
		J^{(T)}(t, \widetilde X)\geq \mathcal X^\top_{t-} A^{(T)}_t\mathcal X_{t-} +\mathcal X^\top_{t-} B^{(T)}_t+C^{(T)}_t.
	\end{equation*}	
	On the other hand , we can easily check that the state processes given by \eqref{state-T-X} and \eqref{state-T-Y} satisfy the equality $I^{A,(T)}  {\mathcal X}^{*,(T)}+I^{B,(T)}\equiv0$ with $\mathcal X^{*,(T)}:=(\widetilde X^{*,(T)},\widetilde Y^{*,(T)})^\top$, which implies that
	\begin{equation*}\label{alt}
		J^{(T)}(t, \widetilde X^{*,(T)})= \mathcal X^\top_{t-} A^{(T)}_t\mathcal X_{t-} +\mathcal X^\top_{t-} B^{(T)}_t+C^{(T)}_t.
	\end{equation*}	
	Therefore, $\widetilde X^{*,(T)}$ defined in \eqref{state-T-X} is indeed optimal if it is a semimartingale.
	
	Uniqueness follows from the strict convexity of $J^{(T)}$; refer to \eqref{cost-T-2}.
\end{proof}

\section{Heuristic derivation of BSDEs}\label{sec:heuristics}
In this section, we motivate the BSDEs characterizing the value function and optimal strategy by considering heuristic limits of discrete time models. 
\subsection{The model with infinite horizon }
Assume trading occurs at $n\Delta$, $n=0,1,\cdots$ Let $n\Delta=t$. 
Given $\lim_{\Delta\rightarrow 0}(1-\Delta\phi)^n=\lim_{\Delta\rightarrow 0} e^{n\log(  1-\Delta\phi  )}=\lim_{\Delta\rightarrow 0}  e^{ \frac{\log(1-\Delta\phi)}{\Delta}t  }=e^{-\phi t}$, the continuous time model can be discretized as\footnote{Here, $Y_{n\Delta-}$ should be understood as the price deviation after the arrivial of the external flow but before the arrival of the internal flow. That is, $Y_{n\Delta-}=Y^{\textrm{before}}_{n\Delta-}+Y^{e}_{n\Delta-}$, where $Y^{\textrm{before}}_{n\Delta-}$ is the deviation before the arrival of both the external and the internal flows at time $n\Delta$, and $Y^e_{n\Delta-}$ is the deviation caused by the external flow at $n\Delta$. This is consistent with the "no-frontrunning" assumption in \cite{MuhleKarbe}. Indeed, the same as \cite{FHX-2023} and \cite{MuhleKarbe}, linearity of the price deviation yields the decomposition $Y=Y^e+Y^o$, where the deviations caused by the external flow and the internal flow follow $dY^e_t=-\rho_t Y^e_t\,dt+\sigma_t\,dW_t$, respectively, $dY^o_t=-\rho_tY^o_t\,dt+\gamma\,dX_t$. In the discrete time model, the no-frontrunning assumption yields the trading profit $\sum_n(1-\phi\Delta)^n( Y^{\textrm{before}}_{n\Delta-} + \Delta Y^e_{n\Delta}+\frac{1}{2}\Delta Y^o_{n\Delta} )\Delta X_{n\Delta}$, where we assume the order is executed at the mid-price. One can verify directly that the heuristic limit of the sum correponds to the first three terms in \eqref{cost-inf-continuous}. }
\begin{equation}\label{cost-inf-discrete}
	J(\xi) = \mathbb E\left[  \sum_{n=0}^\infty (1-\Delta\phi)^n \left(    Y_{n\Delta-}\xi_n + \frac{\gamma}{2}\xi^2_n+\lambda_{n\Delta}\Delta X^2_{n\Delta-}       \right)     \right],
\end{equation}
such that
\begin{equation}\label{state-inf-discrete}
	\left\{\begin{split}
		X_{(n+1)\Delta-} =&~ X_{n\Delta-} - \xi_n \\
		Y_{(n+1)\Delta-} = &~ (1-\Delta\rho_{n\Delta})( Y_{n\Delta-} + \gamma\xi_n   ) + \sigma_{n\Delta} \epsilon_{n+1},
	\end{split}\right. 
\end{equation}
where $\xi_n$ is the traded volume at $n\Delta$, and $\{\epsilon_n\}$ is a sequence of i.i.d. r.v.s. following $N(0,\Delta)$. 

Consider the transformation
\begin{equation*}
	\left\{\begin{split}
		(1-\Delta\phi)^{\frac{n}{2}} X_{n\Delta} = &~ \widetilde X_{n\Delta}\\
		(1-\Delta\phi)^{\frac{n}{2}} Y_{n\Delta} =&~ \widetilde Y_{n\Delta}\\
		(1-\Delta\phi)^{\frac{n}{2}}\xi_n=&~\widetilde\xi_n.
	\end{split}\right. 
\end{equation*}
Then \eqref{cost-inf-discrete} and \eqref{state-inf-discrete} can be rewritten as
\begin{equation}\label{cost-inf-discrete-tild}
	J(\widetilde\xi) = \mathbb E\left[   \sum_{n=0}^\infty  \widetilde Y_{n\Delta-}\widetilde \xi_{n}+\frac{\gamma}{2} \widetilde\xi^2_n+\lambda_{n\Delta}\Delta\widetilde X^2_{n\Delta-} 				       \right],
\end{equation}
respectively,
\begin{equation}\label{state-inf-discrete-tild}
	\left\{\begin{split}
		\widetilde X_{(n+1)\Delta-} =&~ (1-\Delta\phi)^{\frac{1}{2}}(	\widetilde X_{n\Delta-}-\widetilde\xi_n	)\\
		\widetilde Y_{(n+1)\Delta-} = &~(1-\Delta\rho_{n\Delta})(  1-\Delta\phi  )^{\frac{1}{2}} ( \widetilde Y_{n\Delta-}+\gamma\widetilde\xi_n  )+(  1-\Delta\phi  )^{\frac{n+1}{2}}\sigma_{n\Delta}\epsilon_{n+1}.
	\end{split}\right.
\end{equation}

\subsection{The model with finite horizon}

Consider the finite horizon problem where trading only occurs at $n\Delta$, $n=0,1,\cdots,N$ and $X_{N\Delta}=0$. Let the cost functional be denoted by $J^{(N)}$. Then
\[
J^{(N)}(\xi) = \mathbb E\left[   \sum_{n=0}^N   \widetilde Y_{n\Delta-}\widetilde \xi_{n}+\frac{\gamma}{2} \widetilde\xi^2_n+\lambda_{n\Delta}\Delta\widetilde X^2_{n\Delta-} 					     \right]. 
\]
Let $\mathcal X = ( \widetilde X,\widetilde Y  )^\top$. Then
\begin{equation}
	J^{(N)}(\xi) = \mathbb E\left[ \sum_{n=0}^N \mathcal L^\top \mathcal X_{n\Delta-}\widetilde\xi_n+\mathcal R\widetilde \xi^2_n+\mathcal X_{n\Delta-}^\top\mathcal Q_{n\Delta}\mathcal X_{n\Delta-}        \right]
\end{equation}
and
\begin{equation}\label{state-finite-discrete-vector}
	\mathcal X_{(n+1)\Delta-}= \mathcal A_{n\Delta}\mathcal X_{n\Delta-}+\mathcal B_{n\Delta}\widetilde\xi_n+\mathcal D_{n\Delta}\epsilon_{n+1},
\end{equation}
where 
\[
\mathcal L=\begin{pmatrix}	0 & 1		\end{pmatrix}^\top, \quad \mathcal R = \frac{\gamma}{2},\quad \mathcal Q_{n\Delta} = \begin{pmatrix}  \lambda_{n\Delta}\Delta & 0\\ 0 & 0     \end{pmatrix}, 
\] 
and
\[
\mathcal A_{n\Delta} = (1-\Delta\phi)^{\frac{1}{2}}\begin{pmatrix}   1& 0\\ 0 & 1-\Delta\rho_{n\Delta}		     \end{pmatrix},\quad \mathcal B_{n\Delta} = (1-\Delta\phi)^{\frac{1}{2}}  \begin{pmatrix}     -1 \\ (1-\Delta\rho_{n\Delta})\gamma \end{pmatrix},\quad \mathcal D_{n\Delta} = (1-\Delta\phi)^{ \frac{n+1}{2} }\begin{pmatrix}   0\\  \sigma_{n\Delta} \end{pmatrix}.
\]
By DPP, it holds that given the initial state $\mathcal X_{n\Delta-}$
\begin{equation}\label{DPP-1}
	V_n(\mathcal X_{n\Delta-}		) = \inf_{\widetilde\xi_n} \mathbb E_n\left[   \mathcal L^\top \mathcal X_{n\Delta-}\widetilde\xi_n + \mathcal R\widetilde\xi^2_n+ \mathcal X^\top_{n\Delta-}\mathcal Q_{n\Delta}\mathcal X_{n\Delta-} +	V_{n+1}( \mathcal X_{(n+1)\Delta-}^{(\widetilde\xi)}  )		   \right],
\end{equation}
where the superscript in $\mathcal X^{(\widetilde\xi)}$ is used to emphasize the strategy in \eqref{state-finite-discrete-vector} is $\widetilde\xi$, and $\mathbb E_n$ denotes the conditional expectation given the randomness up to time $n\Delta$.

For a $\delta>0$, let $\widetilde\xi_n=\widehat\xi_n+\delta$ and $\widehat{\mathcal X}_{n\Delta-}=\mathcal X_{n\Delta-}+\mathcal X^o$, where $\mathcal X^o:=\begin{pmatrix}    -\delta  & \gamma\delta			  \end{pmatrix}^\top$. Then it holds that 
\begin{equation*}
	\begin{split}
		\widehat{\mathcal X}^{\widehat\xi}_{(n+1)\Delta-} :=&~ \mathcal A_{n\Delta} \widehat{\mathcal X}_{n\Delta-}+\mathcal B_{n\Delta} \widehat\xi_n + \mathcal D_{n\Delta}\epsilon_{n+1}\\
		=&~ \mathcal A_{n\Delta}( \mathcal X_{n\Delta-}+\mathcal X^o	)+\mathcal B_{n\Delta}( \widetilde\xi_-\delta ) + \mathcal D_{n\Delta}\epsilon_{n+1}\\
		=&~\mathcal A_{n\Delta}\mathcal X_{n\Delta-}+\mathcal B_{n\Delta}\widetilde\xi_n + \mathcal D_{n\Delta}\epsilon_{n+1} +\mathcal A_{n\Delta}\mathcal X^o-\delta\mathcal B_{n\Delta}\\
		=&~ \mathcal X_{(n+1)\Delta-}^{(\widetilde\xi)}.
	\end{split} 
\end{equation*}
Thus, from \eqref{DPP-1} we have
\begin{equation}\label{DPP-2}
	\begin{split}
		V_n( \mathcal X_{n\Delta-} ) =&~ \inf_{\widetilde\xi_n}\mathbb E_n\left[  \mathcal L^\top(   \widehat{\mathcal X}_{n\Delta-}  -\mathcal X^o	  ) (\widehat\xi_n+\delta)  +\mathcal R(\widehat\xi_n+\delta)^2    +(\widehat{\mathcal X}_{n\Delta-}-\mathcal X^o	)^\top\mathcal Q_{n\Delta}(   \widehat{\mathcal X}_{n\Delta-}-\mathcal X^o     ) +V_{n+1}( \widehat{\mathcal X}^{(\widehat\xi)}_{(n+1)\Delta-} )\right]\\
		=&~\inf_{\widehat\xi_n} \mathbb E_n\left[		\mathcal L^\top\widehat{\mathcal X}_{n\Delta-}\widehat\xi_n+\mathcal R\widehat\xi^2_n+\widehat{\mathcal X}^\top_{n\Delta-}\mathcal Q_{n\Delta}\widehat{\mathcal X}_{n\Delta-} + V_{n+1}( \widehat{\mathcal X}_{(n+1)\Delta-}^{(\widehat\xi)} ) 			\right]\\
		&~+\delta\mathcal L^\top\widehat{\mathcal X}-\delta\mathcal L^\top\mathcal X^o+\delta^2\mathcal R-2(\mathcal X^o)^\top\mathcal Q_{n\Delta}\widehat{\mathcal X}_{n\Delta-}+(\mathcal X^o)^\top\mathcal Q_{n\Delta}\mathcal X^o\\
		=&~V_n(\widehat{\mathcal X}_{n\Delta-})  +\delta\mathcal L^\top\widehat{\mathcal X}_{n\Delta-}-\delta\mathcal L^\top\mathcal X^o+\delta^2\mathcal R-2(\mathcal X^o)^\top\mathcal Q_{n\Delta}\widehat{\mathcal X}_{n\Delta-}+(\mathcal X^o)^\top\mathcal Q_{n\Delta}\mathcal X^o. 
	\end{split}
\end{equation}
Consider the ansatz 
\begin{equation}\label{ansatz-discrete}
	V_n(\mathcal X) = \mathcal X^\top A_n\mathcal X+B_n^\top \mathcal X+C_n,
\end{equation}
where $A_n$ is $\mathbb R^{2\times 2}$-valued, $B_n$ is $\mathbb R^2$-valued and $C_n$ is $\mathbb R$-valued. Then 
\begin{equation}\label{DPP-3}
	\begin{split}
		V_n(\widehat{\mathcal X}_{n\Delta-}) =&~ \widehat{\mathcal X}^\top_{n\Delta-} A_n\widehat{\mathcal X}_{n\Delta-}+B_n^\top \widehat{\mathcal X}_{n\Delta-}+C_n\\
		=&~ ({\mathcal X}_{n\Delta-}+\mathcal X^o)^\top A_n ({\mathcal X}_{n\Delta-}+\mathcal X^o)+B_n^\top  ({\mathcal X}_{n\Delta-}+\mathcal X^o)+C_n\\
		=&~ V_n(\mathcal X_{n\Delta-}) + (\mathcal X^o)^\top A_n\mathcal X^o +2(\mathcal X^o)^\top A_n\mathcal X_{n\Delta-} + B^\top_n\mathcal X^o.
	\end{split}
\end{equation}
From \eqref{DPP-2} and \eqref{DPP-3}, we have 
\begin{equation*}
	\begin{split}
		&~\delta\mathcal L^\top\widehat{\mathcal X}_{n\Delta-}-\delta\mathcal L^\top\mathcal X^o+\delta^2\mathcal R-2(\mathcal X^o)^\top\mathcal Q_{n\Delta}\widehat{\mathcal X}_{n\Delta-}+(\mathcal X^o)^\top\mathcal Q_{n\Delta}\mathcal X^o\\
		=&~ -  (\mathcal X^o)^\top A_n\mathcal X^o  - 2(\mathcal X^o)^\top A_n\mathcal X_{n\Delta-} - B^\top_n\mathcal X^o,
	\end{split}
\end{equation*}
or equivalently,
\begin{equation*}
	\begin{split}
		&~O(\delta^2)-2\delta\begin{pmatrix}  -1 & \gamma    \end{pmatrix}A_n\mathcal X_{n\Delta-}-\delta B^\top_n\begin{pmatrix}-1 & \gamma\end{pmatrix}^\top \\
		=&~ \delta \begin{pmatrix}  0 & 1   \end{pmatrix}\mathcal X_{n\Delta-}-2\delta\begin{pmatrix}    -\Delta\lambda_{n\Delta} & 0 \end{pmatrix}\mathcal X_{n\Delta-}.
	\end{split}
\end{equation*}
Divided by $\delta$ on both sides and letting $\delta\rightarrow 0$, we have 
\[
-2\begin{pmatrix}  	 -1  & \gamma		    \end{pmatrix} A_n\mathcal X_{n\Delta-} - B^\top_n\begin{pmatrix}  -1 & \gamma    \end{pmatrix}^\top = \begin{pmatrix}      0 & 1 \end{pmatrix}\mathcal X_{n\Delta-} -2\begin{pmatrix}  -\Delta\lambda_{n\Delta}  & 0    \end{pmatrix}\mathcal X_{n\Delta-}.
\]
Compare the coefficients on $\mathcal X_{n\Delta-}$ and the constant terms, one has 
\begin{equation*}
	\left\{\begin{split}
		-2\begin{pmatrix}  	 -1  & \gamma		    \end{pmatrix} A_n =&~ \begin{pmatrix}    2\Delta\lambda_{n\Delta} & 1 \end{pmatrix}\\
		\begin{pmatrix}		-1 & \gamma	\end{pmatrix} B_n=&~0,
	\end{split}\right.
\end{equation*}
which implies that
\begin{equation}\label{relation:AB-discrete}
	\left\{\begin{split}
		A_{11,n}=&~\gamma A_{12,n}+\Delta\lambda_{n\Delta},\\ 
		A_{12,n} = &~ \gamma A_{22,n}+\frac{1}{2},\\
		B_{1,n}=&~\gamma B_{2,n}.
	\end{split}\right. 
\end{equation}
By the ansatz \eqref{ansatz-discrete}, 
\begin{equation}\label{Vn+1-1}
	\begin{split}
		V_{n+1}\left(\mathcal X^{(\widetilde\xi)}_{ (n+1)\Delta- }		\right) = &~ (\mathcal X^{(\widetilde\xi)}_{ (n+1)\Delta-})^\top A_{n+1}\mathcal X^{(\widetilde\xi)}_{ (n+1)\Delta- }+B_{n+1}^\top \mathcal X^{(\widetilde\xi)}_{ (n+1)\Delta- }+C_{n+1}\\
		=&~ \left(\mathcal A_{n\Delta}\mathcal X_{n\Delta-}+\mathcal B_{n\Delta}\widetilde\xi_n+\mathcal D_{n\Delta}\epsilon_{n+1} \right)^\top A_{n+1}\left(  	\mathcal A_{n\Delta}\mathcal X_{n\Delta-}+\mathcal B_{n\Delta}\widetilde\xi_n+\mathcal D_{n\Delta}\epsilon_{n+1}			     \right)\\
		&~ + B_{n+1}^\top \left(   \mathcal A_{n\Delta}\mathcal X_{n\Delta-}+\mathcal B_{n\Delta}\widetilde\xi_n+\mathcal D_{n\Delta}\epsilon_{n+1}       \right) + C_{n+1}.
	\end{split}
\end{equation}
Taking \eqref{Vn+1-1} into \eqref{DPP-1}, one has 
\begin{equation}\label{DPP-4}
	\begin{split}
		V_n(\mathcal X_{n\Delta-}		) =&~ \inf_{\widetilde\xi_n} \mathbb E_n\left[   \mathcal L^\top \mathcal X_{n\Delta-}\widetilde\xi_n + \mathcal R\widetilde\xi^2_n+ \mathcal X^\top_{n\Delta-}\mathcal Q_{n\Delta}\mathcal X_{n\Delta-} \right.\\
		&~\left.+	\left(\mathcal A_{n\Delta}\mathcal X_{n\Delta-}+\mathcal B_{n\Delta}\widetilde\xi_n+\mathcal D_{n\Delta}\epsilon_{n+1} \right)^\top A_{n+1}\left(  	\mathcal A_{n\Delta}\mathcal X_{n\Delta-}+\mathcal B_{n\Delta}\widetilde\xi_n+\mathcal D_{n\Delta}\epsilon_{n+1}			     \right) \right.\\
		&~\left. + B_{n+1}^\top \left(   \mathcal A_{n\Delta}\mathcal X_{n\Delta-}+\mathcal B_{n\Delta}\widetilde\xi_n+\mathcal D_{n\Delta}\epsilon_{n+1}       \right) + C_{n+1}  \right]\\
		=&~ \inf_{\widetilde\xi^n}\left\{         \mathcal L^\top \mathcal X_{n\Delta-}\widetilde\xi_n + \mathcal R\widetilde\xi^2_n+ \mathcal X^\top_{n\Delta-}\mathcal Q_{n\Delta}\mathcal X_{n\Delta-}  \right.\\
		&~\left.   +\mathbb E_n\left[ \left(\mathcal A_{n\Delta}\mathcal X_{n\Delta-}+\mathcal B_{n\Delta}\widetilde\xi_n  \right)^\top A_{n+1}\left(  	\mathcal A_{n\Delta}\mathcal X_{n\Delta-}+\mathcal B_{n\Delta}\widetilde\xi_n      \right)  \right]    \right.\\
		&~+ 2\mathbb E_n\left[(\mathcal A_{n\Delta}\mathcal X_{n\Delta-} + \mathcal B_{n\Delta}\widetilde\xi_n		)^\top A_{n+1} \mathcal D_{n\Delta}\epsilon_{n+1}\right]	\\
		&~\left. + \mathbb E_n\left[B_{n+1}^\top \left(   \mathcal A_{n\Delta}\mathcal X_{n\Delta-}+\mathcal B_{n\Delta}\widetilde\xi_n      \right)\right]  \right\}\\
		&~+\mathbb E_n\left[	\mathcal D_{n\Delta}^\top A_{n+1}\mathcal D_{n\Delta}\epsilon_{n+1}^2 	+B_{n+1}^\top \mathcal D_{n\Delta}\epsilon_{n+1} + C_{n+1}	\right].
	\end{split}
\end{equation}
The minimization in \eqref{DPP-4} implies that 
\begin{equation}\label{xi-star-discrete}
	\begin{split}
		\widetilde\xi_n^*=&~-\left(2\mathcal R+2\mathcal B_{n\Delta}^\top\mathbb E_n[ A_{n+1} ]\mathcal B_{n\Delta} 		\right)^{-1} \Big\{     \left(  \mathcal L^\top+    2\mathcal B^\top_{n\Delta}\mathbb E_n[A_{n+1}]\mathcal A_{n\Delta}   \right) \mathcal X_{n\Delta-} 	 \\
		&~  \qquad\qquad \qquad \qquad \qquad\qquad \qquad\quad + 	 2\mathcal B_{n\Delta}^\top \mathbb E_n[A_{n+1}\epsilon_{n+1}]\mathcal D_{n\Delta} + \mathbb E_n[	B_{n+1}		]^\top\mathcal B_{n\Delta}     \Big\}.
	\end{split}
\end{equation}
Taking \eqref{xi-star-discrete} into \eqref{DPP-4}, we get
\begin{align*}
%	\begin{split}
		&~V_n(\mathcal X_{n\Delta-})\\
		= &~-\frac{1}{4} \left( \mathcal R+\mathcal B^\top_{n\Delta}\mathbb E_n[A_{n+1}]\mathcal B_{n\Delta}		\right)^{-1}\Big\{   	(  \mathcal L^\top+2\mathcal B_{n\Delta}^\top \mathbb E_n[A_{n+1}]\mathcal A_{n\Delta}    )\mathcal X_{n\Delta-}			\\
		&~\qquad\qquad \qquad\qquad\qquad\qquad\qquad \quad  + 2\mathcal B_{n\Delta}^\top\mathbb E_n[A_{n+1}\epsilon_{n+1}]\mathcal D_{n\Delta}+\mathbb E_n[	B_{n+1}	]^\top\mathcal B_{n\Delta }  \Big\}^2\\
		&~+\mathcal X_{n\Delta-}^\top\mathcal Q_{n\Delta}\mathcal X_{n\Delta-} + \mathcal X^\top_{n\Delta-}\mathcal A^\top_{n\Delta}\mathbb E_n[A_{n+1}] \mathcal A_{n\Delta}\mathcal X_{n\Delta-}+2\mathcal X_{n\Delta-}^\top \mathcal A^\top_{n\Delta-}\mathbb E_n[ A_{n+1}\epsilon_{n+1} ]\mathcal D_{n\Delta} \\
		&~+ \mathbb E_n[ B_{n+1} ]^\top\mathcal A_{n\Delta}\mathcal X_{n\Delta-} + \mathcal D^\top_{n\Delta}\mathbb E_n[ A_{n+1}\epsilon^2_{n+1} ]\mathcal D_{n\Delta} + \mathbb E_n[B_{n+1}]^\top\mathcal D_{n\Delta}\epsilon_{n+1} +\mathbb E_n[C_{n+1}] \\
		=&~ \mathcal X^\top_{n\Delta-} \Big\{    -\frac{1}{4}\left( \mathcal R+\mathcal B^\top_{n\Delta}\mathbb E_n[A_{n+1}]\mathcal B_{n\Delta}		\right)^{-1} ( \mathcal L^\top+2\mathcal B^\top_{n\Delta}\mathbb E_n[A_{n+1}]\mathcal A_{n\Delta}   )^\top (	\mathcal L^\top+2\mathcal B^\top_{n\Delta}\mathbb E_n[A_{n+1}]\mathcal A_{n\Delta} 	) \\
		&~\qquad\quad   +\mathcal Q_{n\Delta}+\mathcal A^\top_{n\Delta}\mathbb E_n[A_{n+1}]\mathcal A_{n\Delta}						       \Big\}\mathcal X_{n\Delta-} \\
		&~+ \Big\{    -\frac{1}{2}\left(    \mathcal R+\mathcal B^\top_{n\Delta}\mathbb E_n[A_{n+1}]\mathcal B_{n\Delta}	     \right)^{-1}	(  2\mathcal B_{n\Delta}^\top\mathbb E_n[A_{n+1}\epsilon_{n+1}]\mathcal D_{n\Delta}+\mathbb E_n[ B_{n+1} ]^\top\mathcal B_{n\Delta}   )	(  \mathcal L^\top+2\mathcal B_{n\Delta}^\top\mathbb E_n[A_{n+1}]\mathcal A_{n\Delta}  )	\\
		&~\qquad + 	2\mathcal D_{n\Delta}^\top \mathbb E_n[A_{n+1}\epsilon_{n+1}]\mathcal A_{n\Delta}  + \mathbb E_n[ B_{n+1}		]^\top\mathcal A_{n\Delta}			    \Big\} \mathcal X_{n\Delta-}\\
		&~-\frac{1}{4}\left(    \mathcal R+\mathcal B^\top_{n\Delta}\mathbb E_n[A_{n+1}]\mathcal B_{n\Delta}	     \right)^{-1} \left(	 2\mathcal B_{n\Delta}^\top\mathbb E_n[A_{n+1}\epsilon_{n+1}]\mathcal D_{n\Delta} + \mathbb E_n[B_{n+1}]^\top\mathcal B_{n\Delta}				\right)^2\\
		&~+\mathcal D_{n\Delta}^\top\mathbb E_n[A_{n+1}\epsilon^2_{n+1}]\mathcal D_{n\Delta}+ \mathbb E_n[ B_{n+1}\epsilon_{n+1} ]^\top\mathcal D_{n\Delta}+\mathbb E_n[C_{n+1}]\\
		=&~\mathcal X^\top_{n\Delta-} A_n\mathcal X_{n\Delta-}+B_n^\top \mathcal X_{n\Delta-}+C_n.
%	\end{split}
\end{align*}
By comparing the coefficients of $\mathcal X^\top_{n\Delta-}(\cdots)\mathcal X_{n\Delta-}$, $(\cdots)\mathcal X_{n\Delta-}$ and other terms, we have 
\begin{equation}
	\left\{\begin{split}
		A_n=&~\mathcal Q_{n\Delta} + \mathcal A^\top_{n\Delta}\mathbb E_n[ A_{n+1} ]\mathcal A_{n\Delta}-   a_n^{-1}(I_n^A)^\top I_n^A\\
		B_n^\top=&~ -2a_n^{-1}I_n^B I_n^A+\mathbb E_n[ B_{n+1} ]^\top\mathcal A_{n\Delta} + 2\mathcal D_{n\Delta}^\top \mathbb E_n[A_{n+1}\epsilon_{n+1}]\mathcal A_{n\Delta }\\
		C_n=&~ -a_n^{-1}(I^B_n)^2 + \mathcal D_{n\Delta}^\top\mathbb E_n[ A_{n+1}\epsilon^2_{n+1} ]\mathcal D_{n\Delta} + \mathbb E_n[ B_{n+1}\epsilon_{n+1} ]^\top\mathcal D_{n\Delta} +\mathbb E_n[C_{n+1}],
	\end{split}\right. 
\end{equation}
where 
\begin{equation*}
	\left\{\begin{split}
		a_n=&~\mathcal R+\mathcal B_{n\Delta}^\top\mathbb E_n[ A_{n+1} ] \mathcal B_{n\Delta},\\
		I_n^A = &~ \frac{1}{2}\mathcal L^\top+\mathcal B_{n\Delta}^\top\mathbb E_n[A_{n+1}]\mathcal A_{n\Delta},\\
		I_n^B=&~\frac{1}{2}\mathbb E_n[B_{n+1}]^\top \mathcal B_{n\Delta} + \mathcal B^\top_{n\Delta}\mathbb E_n[ A_{n+1}\epsilon_{n+1} ]\mathcal D_{n\Delta}. 
	\end{split}\right. 
\end{equation*}
Let $\mathcal X:=(x,y)^\top$ be the initial position at $N$, i.e. $(X_{N\Delta-},Y_{N\Delta-})=(x,y)$. Then the optimal strategy at $N\Delta$ is $\xi_N=x$ to close the position immediately. The cost due to this strategy is $(1-\Delta\phi)^N(xy+\frac{\gamma}{2}x^2+\lambda_{N\Delta}\Delta x^2)=\widetilde x\widetilde y+\frac{\gamma}{2}\widetilde x^2+\lambda_{N\Delta}\Delta\widetilde x^2=\begin{pmatrix}\widetilde x & \widetilde y    \end{pmatrix} \begin{pmatrix}  \frac{\gamma}{2}+\lambda_{N\Delta}\Delta & \frac{1}{2} \\  \frac{1}{2}  & 0    \end{pmatrix}\begin{pmatrix}  \widetilde x\\ \widetilde y   \end{pmatrix}$. Thus, the terminal conditions can be determined
\[
A_{N\Delta} = \begin{pmatrix}  \frac{\gamma}{2}+\lambda_{N\Delta}\Delta & \frac{1}{2} \\  \frac{1}{2}  & 0    \end{pmatrix},\quad B_{N\Delta} = \begin{pmatrix}  0 \\ 0  \end{pmatrix},\quad C_{N\Delta}=0. 
\]

We conjecture that the limits of $A_n$, $B_n$ and $C_n$ follow BSDEs. Thus, it is sufficient to motivate the driver. Specifically, 
\begin{equation*}
	\begin{split}
		\mathbb E_n\left[\frac{A_{n}-A_{n+1}}{\Delta}\right] = \widetilde{\mathcal Q}_{n\Delta}+\frac{	\widetilde{\mathcal A}_{n\Delta}\mathbb E_n[A_{n+1}]\widetilde{ \mathcal A}_{n\Delta} -\mathbb E_n[  A_{n+1}  ]			}{\Delta}  - \phi \widetilde{\mathcal A}_{n\Delta}\mathbb E_n[ A_{n+1} ]\widetilde{\mathcal A}_{n\Delta} -\frac{(I_n^A)^\top}{\Delta} \frac{\Delta}{a_n} \frac{ I_n^A }{\Delta},
	\end{split}
\end{equation*}
where 
\[
\widetilde{\mathcal Q}_{n\Delta}=\begin{pmatrix}  \lambda_{n\Delta} & 0 \\ 0 & 0    \end{pmatrix},\qquad \widetilde{\mathcal A}_{n\Delta} = \begin{pmatrix} 	1 & 0 \\ 1 & 1-\Delta\rho_{n\Delta}		    \end{pmatrix}.
\]
Straightforward calculation implies that the driver of $A$ follows
\[
-\dot A_t  =  \widetilde{\mathcal Q}_t+\mathcal H^\top_t A_t+A_t\mathcal H_t -\phi A_t-\frac{(I_t^A)^\top I_t^A}{a_t},
\]
where 
\[
a_t=\gamma\rho_t+\lambda_t+\frac{\phi\gamma}{2},\qquad \mathcal H_t=\begin{pmatrix}  0 & 0\\ 0 & -\rho_t  \end{pmatrix},\qquad I^A_t=\begin{pmatrix} -\rho_t A_{11,t}-\lambda_t & -\frac{\rho_t A_{11,t}}{\gamma} +\rho_t+\frac{\phi}{2}     \end{pmatrix}. 
\]
Thus, the limit $A$ satisfies the following BSDE
\begin{equation}
	\left\{\begin{split}
		-dA_t=&~\left\{\widetilde{\mathcal Q}_t+\mathcal H^\top_t A_t+A_t\mathcal H_t -\phi A_t-\frac{(I_t^A)^\top I_t^A}{a_t}\right\}\,dt - Z^A_t\,dW_t\\
		A_T=&~\begin{pmatrix}  \frac{\gamma}{2} & \frac{1}{2} \\  \frac{1}{2}  & 0    \end{pmatrix}.
	\end{split}\right. 
\end{equation}
Moreover, heuristically, we have
\[
\frac{\mathbb E_n[A_{n+1}\epsilon_{n+1}]}{\Delta} = \mathbb E_n\left[\frac{A_{n+1}-A_n }{\Delta}\epsilon_{n+1}\right] \rightarrow Z^A_t.
\]
Thus, we have the following heuristic convergence 
\begin{equation*}
	\begin{split}
		\mathbb E_n\left[  \frac{  B_n^\top- B_{n+1}^\top		}{\Delta}         \right] =&~ -\frac{ \Delta   }{ a_n   } \frac{I_n^B}{\Delta}\frac{I_n^A}{\Delta}  + \frac{\mathbb E_n[  B_{n+1} ]^\top \mathcal A_{n\Delta} - \mathbb E_n[B_{n+1}]^\top     }{\Delta}  + \frac{2\mathcal D_{n\Delta}^\top \mathbb E_n[   A_{n+1}\epsilon_{n+1}   ]\mathcal A_{n\Delta}  }{\Delta}\\
		\rightarrow&~-\frac{2I^B_tI^A_t}{a_t}+ B^\top_t\mathcal H_t-\frac{\phi}{2}B^\top_t + 2\mathcal D_t Z^A_t,
	\end{split}
\end{equation*}
where
\[
I^B_t = -\frac{1}{2}\rho_t B_{1,t}+\begin{pmatrix}    -1 & \gamma    \end{pmatrix}Z^A_t\mathcal D_t=  -\frac{1}{2}\rho_t B_{1,t}   ,\qquad \mathcal D_t = e^{-\frac{\phi t}{2}}\begin{pmatrix}  	0 & \sigma_t 			    \end{pmatrix}^\top. 
\]
The limit $B$ is conjectued to satisfy the following BSDE
\begin{equation}
	\left\{\begin{split}
		-dB^\top_t = &~\left\{-\frac{2I^B_tI^A_t}{a_t}+ B^\top_t\mathcal H_t-\frac{\phi}{2}B^\top_t + 2\mathcal D_t Z^A_t  \right\}\,dt-Z^B_t\,dW_t\\
		B^\top_T=&~\begin{pmatrix}  0 & 0    \end{pmatrix}.
	\end{split}\right.
\end{equation}
Finally, by the heuristic limit
\begin{equation*}
	\begin{split}
		\frac{1}{\Delta}\mathcal D_{n\Delta}^\top\mathbb E_n[ A_{n+1}\epsilon^2_{n+1} ]\mathcal D_{n\Delta} =&~ \mathcal D_{n\Delta}^\top\mathbb E_n\left[ \frac{A_{n+1}-A_n}{\Delta}\epsilon^2_{n+1} \right]\mathcal D_{n\Delta}+  \frac{1}{\Delta}\mathcal D_{n\Delta}^\top\mathbb E_n[ A_{n}\epsilon^2_{n+1} ]\mathcal D_{n\Delta} \\
		\rightarrow&~ \mathcal D_t^\top \langle A,t \rangle_t\mathcal D_t+\mathcal D^\top_t A_t\mathcal D_t\\
		=&~\mathcal D^\top_t A_t\mathcal D_t,
	\end{split}
\end{equation*}
we have 
\begin{equation*}
	\begin{split}
		\mathbb E_n\left[  \frac{	C_n-C_{n+1}	}{\Delta}     \right] =&~ -\frac{\Delta}{a_n} \frac{(I_n^B)^2}{\Delta^2}+\frac{1}{\Delta}\mathcal D_{n\Delta}^\top\mathbb E_n[ A_{n+1}\epsilon^2_{n+1} ]\mathcal D_{n\Delta} + \mathbb E_n\left[ \frac{B_{n+1}-B_n}{\Delta}\epsilon_{n+1} \right]^\top\mathcal D_{n\Delta}\\
		\rightarrow&~ -\frac{1}{a_t} (I^B_t)^2+\mathcal D^\top_t A_t\mathcal D_t+(Z^B_t)^\top\mathcal D_t. 
	\end{split}
\end{equation*}
Thus, the limit $C$ is conjectured to satisfy the following BSDE
\begin{equation}
	\left\{\begin{split}
		-dC_t=&~\left\{-\frac{1}{a_t} (I^B_t)^2+\mathcal D^\top_t A_t\mathcal D_t+(Z^B_t)^\top\mathcal D_t\right\}\,dt-Z^C_t\,dW_t\\
		C_T=&~0.
	\end{split}\right.
\end{equation}

\subsection{The initial jump}

The initial action of the investor is conjectured to be a jump strategy following the heuristic limit:
\begin{equation*} 
	\begin{split}
		\widetilde\xi_n^*=&~-\left(2\mathcal R+2\mathcal B_{n\Delta}^\top\mathbb E_n[ A_{n+1} ]\mathcal B_{n\Delta} 		\right)^{-1} \Big\{     \left(  \mathcal L^\top+    2\mathcal B^\top_{n\Delta}\mathbb E_n[A_{n+1}]\mathcal A_{n\Delta}   \right) \mathcal X_{n\Delta-} 	 \\
		&~  \qquad\qquad \qquad \qquad \qquad\qquad \qquad\quad + 	 2\mathcal B_{n\Delta}^\top \mathbb E_n[A_{n+1}\epsilon_{n+1}]\mathcal D_{n\Delta} + \mathbb E_n[	B_{n+1}		]^\top\mathcal B_{n\Delta}     \Big\}\\
		=&~-\frac{I_n^A}{a_n}\mathcal X_{n\Delta-}-\frac{I_n^B}{a_n}\\
		\rightarrow&~ -\frac{I_t^A}{a_t}\mathcal X_{t-}-\frac{I^B_t}{a_t}.
	\end{split}
\end{equation*}

\bibliography{bib_CFX2023}
%\bibliography{bib_Guanxing_singular}

\end{document}